\documentclass[12pt,journal,draftclsnofoot,onecolumn,english]{IEEEtran}
\IEEEoverridecommandlockouts

\usepackage{cite}
\usepackage{amsmath,amssymb,amsfonts}
\usepackage{epsfig,subfigure,multirow}
\usepackage{algorithmic}
\usepackage{graphicx}
\usepackage{textcomp}
\usepackage{xcolor, color}
\usepackage[pagebackref=false]{hyperref}
\usepackage{array}
\usepackage{pstricks}
\usepackage{enumitem}
\usepackage{verbatim}
\usepackage{stfloats}
\usepackage[bookmarks=false]{}
\usepackage{aecompl}
\usepackage{mathrsfs}
\usepackage{arydshln}
\usepackage{mathtools} 
\usepackage{amsthm}
\usepackage{caption}
\usepackage{lipsum}
\usepackage{bm}
\usepackage{graphicx}
\usepackage{algorithm}
\usepackage{algorithmic}
\usepackage{comment}
\usepackage{booktabs}
\usepackage{tikz}
\usetikzlibrary{calc}
\usetikzlibrary{fit}

\def \BibTeX{{\rm B\kern-.05em{\sc i\kern-.025em b}\kern-.08em
		T\kern-.1667em\lower.7ex\hbox{E}\kern-.125emX}}

\newtheorem{theorem}{Theorem}

\newtheorem{definition}{Definition}

\newtheorem{remark}{Remark}

\newtheorem{claim}{Claim}
\newcommand{\defeq}{\coloneqq}
\begin{document}
	

\title{Hierarchical Gradient Coding: From Optimal Design to Privacy at Intermediate Nodes}
	
    \author{
 Ali~Gholami,~\IEEEmembership{Student Member,~IEEE,}
 Tayyebeh Jahani-Nezhad,~\IEEEmembership{Member,~IEEE,} \\
 Kai~Wan,~\IEEEmembership{Member,~IEEE,} 
and~Giuseppe Caire,~\IEEEmembership{Fellow,~IEEE}
  \thanks{
A short version of this paper   was presented  at   the 2025 IEEE International Symposium on Information Theory (ISIT)~\cite{11195435}. 
}
\thanks{A.~Gholami,  T.~Jahani-Nezhad, and G.~Caire are with the Electrical Engineering and Computer Science Department, Technische Universit\"at Berlin, 10587 Berlin, Germany (e-mail: \{a.gholami, t.jahani.nezhad, caire\}@tu-berlin.de). 
}
\thanks{
K.~Wan is with the School of Electronic Information and Communications,
Huazhong University of Science and Technology, 430074  Wuhan, China,  (e-mail: kai\_wan@hust.edu.cn). 
}

	}
	
    \maketitle
	
    \begin{abstract}
Gradient coding is a distributed computing technique for computing gradient vectors over large datasets by outsourcing partial computations to multiple workers, typically connected directly to the server. In this work, we investigate gradient coding in a hierarchical setting, where intermediate nodes sit between the server and workers.
        This structure reduces the 
        communication load      received   at the server, which is a  bottleneck in conventional gradient coding systems.  In this paper, the intermediate nodes, referred to as \textit{relays}, process the data received from workers and send the results to the server for the final gradient computation. Our main contribution is deriving the optimal communication-computation trade-off by designing a linear coding scheme, 
        also considering straggling and adversarial nodes among both relays and workers. 
        We propose a coding scheme which 
        achieves both the optimal relay-to-server communication load and the optimal worker-to-relay communication load. 
        We further extend our setting to incorporate privacy by requiring that relays learn no information about the computed partial gradients from the messages they receive. 
        This is achieved by introducing shared randomness among workers, allowing each worker to encode its partial gradients such that the randomness cannot be canceled out  at the relay. Meanwhile, the server can successfully decode the global gradient by eliminating this randomness after receiving the computations of the non-straggling relays. Importantly, this privacy guarantee is achieved without increasing the overall communication load.
    \end{abstract}

    \section{Introduction}
    Distributed machine learning has become a vital approach for training complex models on large-scale datasets \cite{dean2012large,ahmed2013distributed,li2014communication,verbraeken2020survey}. In this approach, either the dataset or the model parameters are distributed across multiple computational nodes that collaboratively train or evaluate the model. A common setup, especially for deep neural networks, involves storing the model parameters on a central server while partitioning the dataset among multiple worker nodes. Each worker node then computes partial gradients of the model using its assigned local data and sends the results to the server, which aggregates them to obtain the global gradient. 
    
   Despite its advantages, distributed machine learning faces several key challenges. These include communication overhead due to frequent transmission of high-dimensional gradient vectors \cite{li2014communication, verbraeken2020survey,lin2017deep,agarwal2022utility,lu2018multi,albasyoni2020optimal,chen2020scalecom}; the presence of stragglers, or slow-performing nodes that delay training \cite{dean2013tail,li2018near,ozfatura2020straggler,kadhe2020communication,massny2022nested,krishnan2022sequential}; the computational and storage costs associated with each node \cite{yan2018storage,pinto2018hoard,malandrino2021toward,gaston2013realistic}; privacy concerns \cite{nodehi2018limited, yu2019lagrange}; and the threat of adversarial nodes that can compromise the integrity of the training process \cite{blanchard2017machine,chen2017distributed}. 

To address communication and straggler issues, Tandon {\it et al.}~\cite{tandon2017gradient} introduced a coded computing framework known as \emph{gradient coding}, which allows the server to recover the global gradient while tolerating a fixed number of stragglers with minimal communication in homogeneous systems where all worker nodes have equal computation power and storage resources. Ye {\it et al.}~\cite{ye2018communication} further analyzed the trade-offs among communication cost, computational load, and straggler tolerance. 
Building upon this, Cao {\it et al.}~\cite{cao2021adaptive} proposed an adaptive gradient coding scheme that achieves exact gradient recovery with minimal communication cost, even when the number of stragglers varies in real-time.
The authors in \cite{jahani2021optimal} extended this line of work to heterogeneous settings, characterizing the optimal communication cost under arbitrary data placement and the presence of stragglers and adversaries, where  a universal polynomial function was proposed to construct an achievable  coding scheme. 
This problem was further generalized in \cite{wan2021distributed, wan2021tradeoff}  to scenarios where a user node requests a linearly separable function of the data, requiring worker nodes to compute multiple linear combinations of the partial gradients. 
In parallel to addressing communication, computation, and straggler issues through gradient coding, other  researches have also focused on mitigating privacy risks and defending against adversarial behavior in distributed training. For instance, secure aggregation protocols in federated learning, where each worker node has its own local dataset, aim to ensure that individual updates from worker nodes remain private even from the server during aggregation \cite{
bonawitz2017practical, bell2020secure, so2021turbo, yang2021lightsecagg,jahani2022swiftagg+,zhao2021information, 9276464, 11026105, 11027469}.



Most studies on gradient coding, e.g., \cite{tandon2017gradient, ye2018communication,cao2021adaptive,jahani2021optimal,wan2021distributed, wan2021tradeoff}, have focused on the conventional server-worker star topology, where the server is connected to each worker through an individual link.  However, this architecture for distributed learning requires extensive resources, particularly at the server node, where high network bandwidth usage can create critical bottlenecks, especially in practical applications. 
To alleviate this issue, researchers have explored \emph{hierarchical} architectures, which introduce intermediate nodes between the server and the workers. \cite{prakash2020hierarchical, krishnan2023hierarchical, zhang2025fundamental, lu2024capacity, egger2023private} studied the hierarchical gradient aggregation problem in the context of secure aggregation, where the workers are connected to some helper nodes, and the helper nodes are connected to the server. However, these works address a different problem than gradient coding, typically assuming that each worker has an individual local dataset with no overlap among the datasets of any group of workers.

In \cite{reisizadeh2019tree, bhattacharya2021improved, reisizadeh2021codedreduce, shah2024tree}, the authors studied the tree structure in gradient coding, where for each intermediate node there is a set of children nodes and a parent node. Each intermediate node has the responsibility of aggregating the partial gradients received from the children nodes and its own individual partial gradient and sends the result to the parent node. However, tree-based gradient coding schemes suffer from a high communication load, as they focus only on minimizing the computation cost of each worker node. A recent study on hierarchical gradient coding was proposed in~\cite{tang2024design}, which focuses solely on minimizing the computation load, without addressing the communication load. 

In this paper, we propose an alternative scheme for hierarchical gradient coding, where  a server is connected to a group of \textit{relay} nodes, with each relay node connected to a group of \textit{workers}. Based on the pre-determined arbitrary heterogeneous data assignment, the workers compute the partial gradients for the datasets assigned to them and send   linearly encoded vectors to their respective relays. The relay nodes can proceed with computation on the data received from the workers and send the results  to the server. The server should recover the sum of partial gradients from the messages received from the relays while coping with stragglers and adversarial nodes both among the relays and the workers. 
Our main contribution on this model is as follows.
\begin{itemize}
    \item We characterize the optimal communication-computation trade-off through the design of a linear coding scheme based on polynomial functions, inspired by \cite{jahani2021optimal}. The proposed hierarchical scheme significantly reduces the bandwidth requirement needed by the server compared to non-hierarchical gradient coding scenarios, and this has been made possible by relays doing computations on the received messages from the workers, rather than just forwarding them to the server. 
    \item  Furthermore, in the extended version of our model
 where privacy against relays is considered, we design the workers’ computations such that each includes a random component. These random components are coordinated across the workers in a cluster so that no linear combination of their outputs at the relay can fully eliminate the randomness. 
 This is because the coefficients of the random parts are designed so that the resulting coefficient matrix has full row rank, which has been made possible by carefully designing the randomness assigned to the workers.
 Meanwhile, the server is still able to cancel out the randomness and recover the global gradient via a decoding process. This coordinated structure is made possible through the use of polynomial coded computing techniques.
\end{itemize}

    \noindent{\bf Notation:}
     For $n \in \mathbb{N}$, the notation $[n]$ represents set $ \{1,2,...,n\}$. Vectors and matrices are represented by boldface lowercase and uppercase letters, respectively. 
      $\mathbb{F}_q$  denotes a finite field of size $q$. The cardinality of set $\mathcal{S}$ is denoted by $|\mathcal{S}|$. 
 

    \section{Problem setting} 
    \subsection{Hierarchical Gradient Coding} \label{sec:problemsetting}
    In a standard supervised machine learning setting, we are given a training dataset $\mathcal{D}=\{(\mathbf{x}_i,y_i)\}_{i=1}^M$, where $\mathbf{x}_i$ 
    denotes the $i$-th input sample and $y_i$ its corresponding label, for some integer $M$. The goal is to learn the parameters $\mathbf{\Omega}$ (e.g., of a neural network) by minimizing the empirical loss  $L(\mathcal{D};\mathbf{\Omega})\defeq\frac
    {1}{|\mathcal{D}|}\sum_{i=1}^{M} L(\mathbf{x}_i,y_i;\mathbf{\Omega})$, which is typically optimized using the gradient descent algorithm. The algorithm starts with an initial estimate $\mathbf{\Omega}^{(0)}$; the parameters are then iteratively updated as
    \begin{align}
    \mathbf{\Omega}^{(t+1)}\defeq\mathbf{\Omega}^{(t)}-\eta \mathbf{g}^{(t)},
    \end{align}
    where $\eta>0$ is the learning rate and  $\mathbf{g}^{(t)}\in \mathbb{F}_q^d$ is the gradient of the loss function at iteration $t$ defined as 
    \begin{align}
        \mathbf{g}^{(t)}\defeq\nabla L(\mathcal{D};\mathbf{\Omega}^{(t)})=\frac{1}{M}\sum_{i=1}^M \nabla L(\mathbf{x}_i,y_i;\mathbf{\Omega}^{(t)}),
    \end{align}
    where $\nabla$ indicates the gradient w.r.t $\mathbf{\Omega}^{(t)}$, and $\mathbb{F}_q$ is a finite field of size $q$.

    \begin{figure*}
    \centering
    \resizebox{0.7\columnwidth}{!}{%
    \begin{tikzpicture} [xscale=0.6, yscale=0.8,
    roundnode/.style={circle, draw=black!60, fill=black!5},
    squarednode/.style={rectangle, draw=black!60, fill=black!5},
    ]
        \node[roundnode,inner sep=0.5pt] (1) at (-1.8,0) {\scalebox{0.8}{$W_{1,1}$}};
        \node[roundnode,inner sep=0.5pt] (2) at (-0.5,0) {\scalebox{0.8}{$W_{1,2}$}};
        \node[text width=0.5cm] at (0.7,0) {$\dots$};
        \node[roundnode,inner sep=0.5pt] (3) at (1.8,0) {\scalebox{0.55}{$W_{1,N_2^{(1)}}$}};
        
        \node[roundnode,inner sep=0.5pt] (4) at (3.7,0) {\scalebox{0.8}{$W_{2,1}$}};
        \node[roundnode,inner sep=0.5pt] (5) at (5,0) {\scalebox{0.8}{$W_{2,2}$}};
        \node[text width=0.5cm] at (6.2,0) {$\dots$};
        \node[roundnode,inner sep=0.5pt] (6) at (7.3,0) {\scalebox{0.55}{$W_{2,N_2^{(2)}}$}};

        \node[text width=0.5cm] at (8.7,-1) {$\dots$};
        
        \node[roundnode,inner sep=0.5pt] (7) at (10,0) {\scalebox{0.65}{$W_{N_1,1}$}};
        \node[roundnode,inner sep=0.5pt] (8) at (11.3,0) {\scalebox{0.65}{$W_{N_1,2}$}};
        \node[text width=0.5cm] at (12.5,0) {$\dots$};
        \node[roundnode,inner sep=0.5pt] (9) at (13.6,0) {\scalebox{0.4}{$W_{N_1,N_2^{(N_1)}}$}};

        \node[squarednode] (10) at (0,-2) {Relay 1};
        \node[squarednode] (11) at (5.9,-2) {Relay 2};
        \node[squarednode] (12) at (11.8,-2) {Relay $N_1$};

        \draw[dashed] (-2.5, 0.5) rectangle (2.5,-2.5);
        \draw[dashed] (3, 0.5) rectangle (8,-2.5);
        \draw[dashed] (9.3, 0.5) rectangle (14.3,-2.5);
        
        \draw[-] (1) -- (10);
        \draw[-] (2) -- (10);
        \draw[-] (3) -- (10);

        \draw[-] (4) -- (11);
        \draw[-] (5) -- (11);
        \draw[-] (6) -- (11);

        \draw[-] (7) -- (12);
        \draw[-] (8) -- (12);
        \draw[-] (9) -- (12);

        \node[squarednode] (13) at (6,-4) {Server};
        
        \draw[-] (10) -- (13);
        \draw[-] (11) -- (13);
        \draw[-] (12) -- (13);
        \node at (-0.25, 0.8) {\scalebox{0.9}{Cluster 1}};
        \node at (5.75, 0.8) {\scalebox{0.9}{Cluster 2}};
        \node at (11.75, 0.8) {\scalebox{0.9}{Cluster $N_1$}};
        
    \end{tikzpicture}
    }
    \caption{Hierarchical gradient coding structure.}
    \label{fig:setting}
    \end{figure*}

    In a distributed learning setting, as illustrated in Fig. \ref{fig:setting},  
    we consider a system comprising a central server and multiple worker nodes.
    The training dataset $\mathcal{D}$ is partitioned into $K\in\mathbb{N}$ equal-sized datasets, denoted as $\mathcal{D}=\{\mathcal{D}_1, ...,\mathcal{D}_K\}$.
    In this setting, there are $N_1$ clusters, where Cluster $n$ contains $N_2^{(n)}$ workers. The workers in each cluster are connected to a relay, and the relays are connected to the server. All links are orthogonal and interference-free. The datasets are assigned to the workers in the data placement phase. Then each worker computes the gradients of the assigned datasets using the current model parameter $\mathbf{\Omega}$. 
    Let $\mathbf{g}_k^{(t)} \in \mathbb{F}_q^d$ denote the partial gradient computed from subset $\mathcal{D}_k$ at iteration $t$. Since we focus on a single iteration of gradient descent, we drop the superscript and simply write $\mathbf{g}_k$.
    The ultimate goal of the server is to compute the aggregated partial gradients, i.e.,  $\mathbf{g}\defeq\sum_{k=1}^K \mathbf{g}_k$. 

    Worker $j$ in Cluster $n$  is denoted by $W_{(n,j)}$. The dataset assignment is pre-fixed; that is, we consider arbitrary heterogeneous data assignment as in~\cite{jahani2021optimal}.
    The set of datasets assigned to this worker is denoted by $\Gamma_{(n,j)}$, where $\Gamma_{(n,j)} \neq \emptyset$ by assumption.
 The set of datasets assigned to Relay $n$ (i.e., Cluster $n$), denoted by $\Gamma_n \defeq \bigcup_{j\in [N_2^{(n)}]}\Gamma_{(n,j)}$, is the union of all datasets assigned to workers in Cluster $n$.
 If a dataset is included in $\Gamma_n$, it is assigned to at least one worker in Cluster $n$.
Let $\mathcal{A}_k$ denote the set of clusters to which dataset $\mathcal{D}_k$ is assigned. Dataset $\mathcal{D}_k$ is considered assigned to a cluster if it is assigned to at least one worker within that cluster. The size of this set is denoted by $r_1^k\defeq|\mathcal{A}_k| \leq N_1$. 
If dataset $\mathcal{D}_k$ is assigned to Cluster $n$, then let $\mathcal{A}_k^{(n)}$ denote the set of workers in that cluster who are assigned $\mathcal{D}_k$. The size of this set is denoted by $r_{2,n}^k \defeq |\mathcal{A}_k^{(n)}| \leq N_2^{(n)}$. 
    We define
    \begin{align} 
        r_1 &\defeq \min_{k\in [K]} r_1^k, \label{eq:r}\\
        r_{2,n} &\defeq \min_{\substack{k\in [K] \\ k:\mathcal{D}_{k}\in\Gamma_{n}}} r_{2,n}^k, \label{eq:r2}
    \end{align}
where $r_1$ represents the minimum number of repetitions of any datasets across clusters, and $r_{2,n}$ represents the minimum number of repetitions of datasets across workers inside Cluster $n$.



Each Worker $W_{(n,j)}$ computes the local partial gradients for datasets in $\Gamma_{(n,j)}$ and applies a linear encoding function
       $ \mathcal{E}_{(n,j)}: (\mathbb{F}_q^d)^{|\Gamma_{(n,j)}|} \rightarrow \mathbb{F}_q^{C_2^{(n)}},$
    where $C_2^{(n)}$ denotes the size of the vector sent by each worker in Cluster  $n$ to its relay. The encoded vector computed by $W_{(n,j)}$ is denoted by $\Tilde{\mathbf{g}}_{(n,j)} \in \mathbb{F}_q^{C_2^{(n)}}$;
    \begin{align} \label{eq:workerenc}
        \Tilde{\mathbf{g}}_{(n,j)} \defeq \mathcal{E}_{(n,j)}(\{\mathbf{g}_k\}_{k: \mathcal{D}_k \in \Gamma_{(n,j)}}).
    \end{align}

    Relay $n$ receives the encoded vectors sent by the workers in its cluster, considering at most $s_2^{(n)} \in \mathbb{N}  < N_2^{(n)}$ \textit{stragglers} and $a_2^{(n)} \in \mathbb{N}  < \frac{N_2^{(n)}}{2} $ \textit{adversarial} 
     workers in Cluster $n$. In this setting, the straggling nodes are those with failed links or very delayed communication, and by adversary, we refer to nodes that may send incorrect computations but do not exhibit any other malicious behavior or collude with others. Then using the encoded vectors received from non-straggling workers---whose indices are denoted by the set $\mathcal{F}_n$, such that $(n,j)\in \mathcal{F}_n$ indicates that $W_{(n,j)}$ is a non-straggling worker---Relay $n$ encodes a new message $\Tilde{\mathbf{g}}_{n} \in \mathbb{F}_q^{C_1}$ by using a linear encoding as
    \begin{align} \label{eq:relayenc}
        \Tilde{\mathbf{g}}_{n} \defeq \mathcal{E}_n(\mathcal{F}_n, \{\Tilde{\mathbf{g}}_{(n,j)}\}_{(n,j)\in \mathcal{F}_n}),
    \end{align}
    where 
      $  \mathcal{E}_n: \mathcal{F}_n \times (\mathbb{F}_q^{C_2^{(n)}})^{|\mathcal{F}_n|} \rightarrow \mathbb{F}_q^{C_1}. $

    Relay $n$ sends the encoded vector $\Tilde{\mathbf{g}}_{n}$ to the server. The server tolerates at most $s_1\in\mathbb{N}$ straggling and $a_1\in\mathbb{N}$ adversarial relays. 
     After receiving the encoded vectors from non-straggling relays, whose indices are denoted by the set $\mathcal{F}$, the server recovers the desired gradient vector as
    \begin{align} \label{eq:serverdec}
        \mathbf{g} = \mathcal{H}(\mathcal{F}, \{\Tilde{\mathbf{g}}_{n}\}_{n \in \mathcal{F}}),
    \end{align}
    where 
       $ \mathcal{H}: \mathcal{F} \times (\mathbb{F}_q^{C_1})^{|\mathcal{F}|} \rightarrow \mathbb{F}_q^d$. 

    \begin{definition} [Achievable communication load tuple]
        A communication tuple $(C_1,C_2^{(1)}, ..., C_2^{(N_1)})$ is achievable 
        for any set of straggler and adversarial nodes with parameters componentwise dominated by $(s_1,s_2^{(1)},s_2^{(2)}, \dots,s_2^{(N_1)})$ and $(a_1,a_2^{(1)}, a_2^{(2)},\dots,a_2^{(N_1)})$, respectively, and data replication parameters $r_1,r_{2,1},r_{2,2},\dots,r_{2,N_1}$, 
        if there exists a coding scheme with dataset assignment satisfying $r_1,r_{2,1},r_{2,2},\dots,r_{2,N_1}$, and encoding functions $(\{\mathcal{E}_{(n,j)}\}_{n\in [N_1], j\in [N_2^{(n)}]}, \{\mathcal{E}_{n}\}_{n\in [N_1]})$ and decoding function $\mathcal{H}$, such that the server can correctly recover the desired gradient vector $\mathbf{g}$. The minimum values of $C_1$ and $C_2^{(n)}$ where $n\in [N_1]$ among all achievable schemes are denoted by $C_1^*$ and ${C_2^{(n)}}^*$, respectively. 
    \end{definition}

    \subsection{Privacy-Preserving Hierarchical Gradient Coding} \label{sec:problemsettingwithprivacy}
    In the proposed hierarchical gradient coding model, the presence of relays introduces a privacy risk. While the server distributes the datasets among workers to proceed with computation, the relays are not supposed to learn any information about the datasets or even linear combinations of the computed partial gradients. However, the server must still be able to recover the aggregated gradient vector using the messages received from non-straggling relays. 
    To ensure the privacy of gradient information from the relays, we assume the presence of shared randomness among the workers. This shared randomness enables each worker to mask its local computation such that the relays cannot infer any information about the partial gradients, while still allowing the server to accurately recover the aggregated gradient from the relays' messages.
    


Prior to the computation phase, there is an \emph{offline} phase in which the workers coordinate to agree upon a shared random matrix $\mathbf{Z} \in \mathbb{F}_q^{K' \times d}$. The rows of $\mathbf{Z}$, denoted by $\mathbf{z}_1, \dots, \mathbf{z}_{K'}$, are i.i.d. vectors with entries uniformly drawn from $\mathbb{F}_q$.

Each worker uses a subset of these vectors during the encoding phase, we refer to them as \emph{assigned} random vectors.  For Worker $W_{n,j}$, the set of the random vectors   assigned to this worker is denoted by $\Gamma'_{(n,j)}$. 
 The collection of all assigned random vectors used in Cluster $n$ is denoted by
 $\Gamma'_{n}=\bigcup_{j=1}^{N_2^{(n)}} \Gamma'_{(n,j)}$. 
Let $\mathcal{A'}_k$ denote the set of clusters to which the random vector $\mathbf{z}_k$ is assigned.  $\mathbf{z}_k$ is assigned to Cluster $n$ if $\mathbf{z}_k \in \Gamma'_n$.  The size of this set is denoted by  ${r'}_1^{k}\defeq|\mathcal{A'}_k| \leq N_1$. 
If the random vector $\mathbf{z}_k$ is assigned to Cluster $n$, then the set of workers in that cluster having the random vector $\mathbf{z}_k$ is denoted by $\mathcal{A'}_k^{(n)}$, and its size is denoted by ${r'}_{2,n}^{k}\defeq|\mathcal{A'}_k^{(n)}| \leq N_2^{(n)}$. We define
\begin{align}
    r'_1 &\defeq \min_{k\in [K']} {r'}_1^k, \\
    r'_{2,n} &\defeq \min_{\substack{k\in [K'] \\ k:\mathbf{z}_{k}\in\Gamma_{n}}} {r'}_{2,n}^k,
\end{align}
where $r'_1$ represents the minimum number of repetitions of any random vector across clusters, and $r'_{2,n}$ represents the minimum number of repetitions of random vectors inside Cluster $n$ across workers inside Cluster $n$. 
 
The encoded vector computed by $W_{(n,j)}$ is denoted by $\Tilde{\mathbf{g}}_{(n,j)}^{\pi} \in \mathbb{F}_q^{{C_2^{\pi}}^{(n)}}$, where
    \begin{align} \label{eq:workerencprivacy}
        \Tilde{\mathbf{g}}_{(n,j)}^{\pi} \defeq \mathcal{E}_{(n,j)}^{\pi}(\{\mathbf{g}_k\}_{k: \mathcal{D}_k \in \Gamma_{(n,j)}},\{\mathbf{z}_k\}_{k: \mathbf{z}_k \in \Gamma'_{(n,j)}}),
    \end{align}
    where $\mathcal{E}_{(n,j)}^{\pi}:(\mathbb{F}_q^d)^{|\Gamma_{(n,j)}|} \times (\mathbb{F}_q^d)^{|\Gamma'_{(n,j)}|} \rightarrow \mathbb{F}_q^{{C_2^{\pi}}^{(n)}},$ is a linear encoding function.
    
    Relay $n$ then encodes a new message $\Tilde{\mathbf{g}}_{n}^{\pi} \in \mathbb{F}_q^{C_1^{\pi}}$ by using a linear encoding as
    \begin{align}
        \Tilde{\mathbf{g}}_{n}^{\pi} \defeq \mathcal{E}_n^{\pi}(\mathcal{F}_n, \{\Tilde{\mathbf{g}}_{(n,j)}^{\pi}\}_{(n,j)\in \mathcal{F}_n}),
    \end{align}
    where 
      $  \mathcal{E}_n^{\pi}: \mathcal{F}_n \times (\mathbb{F}_q^{{C_2^{\pi}}^{(n)}})^{|\mathcal{F}_n|} \rightarrow \mathbb{F}_q^{C_1^{\pi}}$,
    and $\mathcal{F}_n$ is the set of indices $(n,j)$ where $W_{(n,j)}$ is a non-straggling worker.    
After receiving the encoded vectors from non-straggling relays, whose indices are denoted by the set $\mathcal{F}$, the server recovers the desired gradient vector as
    \begin{align} \label{eq:serverdec}
        \mathbf{g} = \mathcal{H}(\mathcal{F}, \{\Tilde{\mathbf{g}}_{n}^{\pi}\}_{n \in \mathcal{F}}),
    \end{align}
    where 
       $ \mathcal{H}': \mathcal{F} \times (\mathbb{F}_q^{C_1^{\pi}})^{|\mathcal{F}|} \rightarrow \mathbb{F}_q^d$. 
The privacy criterion is defined as follows: 
for every Relay $n \in [N_1]$, we require
\begin{align} \label{eq:privacy}
    I(\{\Tilde{\mathbf{g}}_{(n,j)}^{\pi}, (n,j)\in \mathcal{F}_n\}; \{\mathbf{g}_k, \mathcal{D}_k\in \Gamma_n\})=0.
\end{align}
This means that  when Relay $n$ receives all the computations $\Tilde{\mathbf{g}}_{(n,j)}^{\pi}$ from non-straggling workers, it does not gain any information about the partial gradients. 
Note that in the privacy-preserving scenario, we assume there are no adversaries among workers, i.e., $a_2^{(1)}=a_2^{(2)}=\dots=a_2^{(N_1)}=0$.


\begin{definition} [Achievable communication load tuple]
    A communication tuple $(C_1^{\pi},{C_2^{\pi}}^{(1)}, ..., {C_2^{\pi}}^{(N_1)})$ is achievable for any set of straggler and adversarial nodes with parameters $(s_1,s_2^{(1)},s_2^{(2)}, \dots,s_2^{(N_1)})$ and $a_1$, respectively, and data replication parameters $r_1,r_{2,1},r_{2,2},\dots,r_{2,N_1}$, if there exists a coding scheme with dataset assignment satisfying $r_1,r_{2,1},r_{2,2},\dots,r_{2,N_1}$, and encoding functions $(\{\mathcal{E}^{\pi}_{(n,j)}\}_{n\in [N_1], j\in [N_2^{(n)}]}, \{\mathcal{E}^{\pi}_{n}\}_{n\in [N_1]})$ and decoding function $\mathcal{H}'$, such that the server can correctly recover the desired gradient vector $\mathbf{g}$ while satisfying the privacy criterion in \eqref{eq:privacy}. The minimum values of $C_1^{\pi}$ and ${C_2^{\pi}}^{(n)}$ where $n\in [N_1]$ among all achievable schemes are denoted by ${C_1^{\pi}}^*$ and ${C_2^{\pi^{(n)}}}^{*}$, respectively. 
\end{definition}

In the following, we first introduce our proposed hierarchical gradient coding scheme without considering privacy at the relays in Section~\ref{sec:secwithout}, and then extend the scheme to the privacy-preserving setting in Section~\ref{sec:settingprivacy}.

\section{Hierarchical Gradient Coding} \label{sec:secwithout}
In this section, we first characterize the optimal communication–computation trade-off for the hierarchical gradient coding setting described in Subsection~\ref{sec:problemsetting}, and present the main results. We then introduce our proposed scheme through a motivating example in Subsection~\ref{sec:exmplwithout}, followed by the general scheme in Subsection~\ref{sec:generalscheme}.

    \subsection{Main Results} \label{sec:results}
    \begin{theorem} \label{thm:main}
        For the hierarchical gradient coding problem defined in Section \ref{sec:problemsetting}, the minimum communication loads are characterized by
        \begin{align}
            C^*_1 = \frac{d}{m_1}, C^{(n)^*}_2 = \frac{d}{m_1m^{(n)}_2} , \forall n\in [N_1], 
        \end{align}
        where $m_1=r_1-2a_1-s_1$ and $m^{(n)}_2=r_{2,n}-2a^{(n)}_2-s^{(n)}_2$.
    \end{theorem}

    \begin{proof}
        The proof of achievability is presented in Section \ref{sec:generalscheme}. 
        
       For the converse, we use the genie-aided approach. We can model the system as consisting of one server and multiple clusters, where each cluster contains a relay with its children workers. In fact, each cluster can be viewed as a ``super worker'', which stores $r_1$ datasets. 
       This transforms the problem into a conventional gradient coding setting with a server connected to multiple workers. By \cite[Theorem 1]{jahani2021optimal}, we have $C^*_1 \geq  \frac{d}{m_1}$.

To show that $C^{(n)}_2 \ge \frac{d}{m_1m^{(n)}_2}, n\in [N_1]$  for any linear encoding and general decoding function, we now focus on Cluster $n$, and demonstrate that the problem within each cluster reduces to a conventional gradient coding setup.

Since the relay performs a linear encoding of the vectors it receives from the workers, the vector sent from Relay $n$ to the server can be written as 
        
        \begin{align}
            \Tilde{\mathbf{g}}_{n} = \mathbf{A}_{C_1 \times |\Gamma_n|d} 
            \begin{bmatrix}
               \vdots \\
               \mathbf{g}_k \\
               \vdots \\
            \end{bmatrix}_{k\in \Gamma_n}.
        \end{align}
        Now if we divide each partial gradient into $m_1$ equally-sized parts as $\mathbf{g}_k=
            \begin{bmatrix}
               \mathbf{g}_k^1 \\
               \vdots \\
               \mathbf{g}_k^{m_1} \\
            \end{bmatrix}$
        (or the zero-padded version of each partial gradient if $m_1 \nmid d$),
        then 
        \begin{align}
            \Tilde{\mathbf{g}}_{n} = \mathbf{A}_{C_1 \times m_1|\Gamma_i|\frac{d}{m_1}} 
            \begin{bmatrix}
               \vdots \\
               \mathbf{g}_k^1 \\
               \vdots \\
               \mathbf{g}_k^{m_1} \\
               \vdots \\
            \end{bmatrix}_{k\in \Gamma_n}.
        \end{align}

        If we consider $C_1=C_1^*=\frac{d}{m_1}$, then this corresponds to a conventional gradient coding scenario with $m_1|\Gamma_n|$ datasets each of size $\frac{d}{m_1}$. Note that this does not change the value of $r_{2,n}$. In such a scenario, the server  wants to compute the aggregated partial  gradient which is of the form
        \begin{align}
             \mathbf{g} = \mathbf{A}'_{\frac{d}{m_1} \times m_1|\Gamma_n|\frac{d}{m_1}} 
             \begin{bmatrix}
               \vdots \\
               \mathbf{g}_k^1 \\
               \vdots \\
               \mathbf{g}_k^{m_1} \\
               \vdots \\
            \end{bmatrix},
        \end{align}
        where
        \begin{align}
            \mathbf{A}'_{\frac{d}{m_1} \times m_1|\Gamma_n|\frac{d}{m_1}}=
            \begin{bmatrix}
               1,0,...,0, & 1,0,...,0, & ... & 1,0,...,0 \\
               0,1,...,0, & 0,1,...,0, & ... & 0,1,...,0 \\
               \vdots & \vdots & ... & \vdots \\
               0,0,...,1, & 0,0,...,1, & ... & 0,0,...,1
            \end{bmatrix}.
        \end{align}
              Since both $\mathbf{A}$ and $\mathbf{A}'$ are full-rank matrices with the same size, we can say that the correspondence exists. 
              Now, given that Cluster $n$ has $s_2^{(n)}$ stragglers and $a_2^{(n)}$ adversaries, the conventional gradient coding result yields
        \begin{align}
            C_2^{(n)}=\frac{d/m_1}{m_2^{(n)}}=\frac{d}{m_1m_2^{(n)}},
        \end{align}
        where $m_2^{(n)}=r_{2,n}-2a_2^{(n)}-s_2^{(n)}$. For any $C_1 > \frac{d}{m_1}$, the problem in each cluster corresponds to more than a single conventional gradient coding (the problem can be divided into at least one conventional gradient coding). So $C_2^*$ occurs when $C_1^*$ occurs. This completes the proof. 
    \end{proof}

    \begin{remark}
        The linearity of the encoding functions is assumed here as a system constraint and this allows us to prove a tight converse. Obviously, our achievability strategy remains valid and yields an achievable load tuple even removing this constraint, but the proof of a tight converse in this case is elusive. Establishing whether the linear encoding function is actually optimal for the general case remains an open problem for future work. 
    \end{remark}

    \subsection{Motivating Example} \label{sec:exmplwithout}
   We then illustrate the main idea of the achievability scheme through an example. For simplicity, here we consider a homogeneous data assignment among workers, where each worker possesses the same number of datasets. Fig.~\ref{fig:diagram exmp} illustrates a hierarchical setting consisting of $N_1=3$ clusters, each containing $3$ workers, i.e., $N_2^{(1)}=N_2^{(2)}=N_2^{(3)}=3$. The dataset is partitioned into $K=9$ equal-sized datasets, and the data assignment is performed such that  $r_1=2$, and  $r_{2,1}=r_{2,2}=r_{2,3}=2$. Note that $r_1=2$ means the number of repetitions of any dataset across all clusters is at least 2, and $r_{2,1}=r_{2,2}=r_{2,3}=2$ means that the number of repetitions of any dataset existing inside any particular cluster is at least $2$ across workers inside the cluster.
    \begin{figure}
    \centering
    \resizebox{0.6\columnwidth}{!}{%
    \begin{tikzpicture} [xscale=0.6, yscale=0.8,
    roundnode/.style={circle, draw=black!60, fill=black!5},
    squarednode/.style={rectangle, draw=black!60, fill=black!5},
    ]
        \node[roundnode,inner sep=0.5pt] (1) at (0.2,0) {\scalebox{0.8}{$W_{1,1}$}};
        \node[roundnode,inner sep=0.5pt] (2) at (1.5,0) {\scalebox{0.8}{$W_{1,2}$}};
        \node[roundnode,inner sep=0.5pt] (3) at (2.8,0) {\scalebox{0.8}{$W_{1,3}$}};
        \node[roundnode,inner sep=0.5pt] (4) at (4.7,0) {\scalebox{0.8}{$W_{2,1}$}};
        \node[roundnode,inner sep=0.5pt] (5) at (6,0) {\scalebox{0.8}{$W_{2,2}$}};
        \node[roundnode,inner sep=0.5pt] (6) at (7.3,0) {\scalebox{0.8}{$W_{2,3}$}};
        \node[roundnode,inner sep=0.5pt] (7) at (9.2,0) {\scalebox{0.8}{$W_{3,1}$}};
        \node[roundnode,inner sep=0.5pt] (8) at (10.5,0) {\scalebox{0.8}{$W_{3,2}$}};
        \node[roundnode,inner sep=0.5pt] (9) at (11.8,0) {\scalebox{0.8}{$W_{3,3}$}};

        \node[squarednode] (10) at (1.5,-2) {Relay 1};
        \node[squarednode] (11) at (6,-2) {Relay 2};
        \node[squarednode] (12) at (10.5,-2) {Relay 3};

        \draw[dashed] (-0.5, 0.5) rectangle (3.5,-2.5);
        \draw[dashed] (4, 0.5) rectangle (8,-2.5);
        \draw[dashed] (8.5, 0.5) rectangle (12.5,-2.5);
        
        \draw[-] (1) -- (10);
        \draw[-] (2) -- (10);
        \draw[-] (3) -- (10);

        \draw[-] (4) -- (11);
        \draw[-] (5) -- (11);
        \draw[-] (6) -- (11);

        \draw[-] (7) -- (12);
        \draw[-] (8) -- (12);
        \draw[-] (9) -- (12);

        \node[squarednode] (13) at (6,-4) {Server};
        
        \draw[-] (10) -- (13);
        \draw[-] (11) -- (13);
        \draw[-] (12) -- (13);
        \node at (1.5, .8) {\scalebox{0.9}{Cluster 1}};
        \node at (6, 0.8) {\scalebox{0.9}{Cluster 2}};
        \node at (10.5, 0.8) {\scalebox{0.9}{Cluster 3}};
        
    \end{tikzpicture}
    }
    \caption{Structure of the motivating example.}
    \label{fig:diagram exmp}
    \end{figure}
    Consider the dataset assignment that is performed as follows, 
    \begin{align*} 
        \Gamma_{(1,1)} &= \{\mathcal{D}_1, \mathcal{D}_3, \mathcal{D}_7, \mathcal{D}_9\},
        \Gamma_{(1,2)} = \{\mathcal{D}_1, \mathcal{D}_3, \mathcal{D}_4, \mathcal{D}_6\}, 
        \Gamma_{(1,3)} &= \{\mathcal{D}_4, \mathcal{D}_6, \mathcal{D}_7, \mathcal{D}_9\}, \\
        \Gamma_{(2,1)} &= \{\mathcal{D}_1, \mathcal{D}_2, \mathcal{D}_7, \mathcal{D}_8\},
        \Gamma_{(2,2)} = \{\mathcal{D}_1, \mathcal{D}_2, \mathcal{D}_4, \mathcal{D}_5\}, 
        \Gamma_{(2,3)} &= \{\mathcal{D}_4, \mathcal{D}_5, \mathcal{D}_7, \mathcal{D}_8\}, \\
        \Gamma_{(3,1)} &= \{\mathcal{D}_2, \mathcal{D}_3, \mathcal{D}_8, \mathcal{D}_9\},
        \Gamma_{(3,2)} = \{\mathcal{D}_2, \mathcal{D}_3, \mathcal{D}_5, \mathcal{D}_6\}, 
        \Gamma_{(3,3)} &= \{\mathcal{D}_5, \mathcal{D}_6, \mathcal{D}_8, \mathcal{D}_9\},
    \end{align*}
    which leads that
    \begin{align*}
     &   \Gamma_{1} = \{\mathcal{D}_1, \mathcal{D}_3, \mathcal{D}_4, \mathcal{D}_6, \mathcal{D}_7, \mathcal{D}_9\}, \\
       & \Gamma_{2} = \{\mathcal{D}_1, \mathcal{D}_2, \mathcal{D}_4, \mathcal{D}_5, \mathcal{D}_7, \mathcal{D}_8\}, \\
       & \Gamma_{3} = \{\mathcal{D}_2, \mathcal{D}_3, \mathcal{D}_5, \mathcal{D}_6, \mathcal{D}_8, \mathcal{D}_9\}.
    \end{align*}
    In this example, we assume that none of the relays are stragglers or adversaries; however, within each cluster, there is a single straggler, i.e., $a_1=s_1=0, a_2^{(1)}=a_2^{(2)}=a_2^{(3)}=0, s_2^{(1)}=s_2^{(2)}=s_2^{(3)}=1$. We then define $m_1$ and $m_2^{(i)}$ as $m_1 \defeq r_1-2a_1-s_1=2, m_2^{(i)}\defeq r_{2,i}-2a_2^{(i)}-s_2^{(i)}=1$, for $i=1,2,3$.
    Divide each partial gradient  into $m_1=2$ equal-sized parts, represented as $\mathbf{g}_k = (\mathbf{g}_k[1], \mathbf{g}_k[2])$. 
    Let us focus on Cluster 1. For this cluster, since $m_2^{(1)}=1$, each of these parts is further divided into $m_2^{(1)}=1$ equal-sized parts, represented as $\mathbf{g}_k[i] = (\mathbf{g}_k[i,1])$. It is evident that $\mathbf{g}_k[i]=\mathbf{g}_k[i,1]$. The workers in this cluster compute a linear combination of the elements of its computed partial gradients as follows, 
    \begin{align*}
        \Tilde{\mathbf{g}}_{(1,1)} &\defeq \frac{8}{9} \mathbf{g}_{1}[1] + \frac{2}{3} \mathbf{g}_{3}[1]  + \frac{2}{3} \mathbf{g}_{7}[1] + \frac{1}{2} \mathbf{g}_{9}[1] \nonumber \\
        &-\frac{1}{3} \mathbf{g}_{1}[2] - \frac{2}{9} \mathbf{g}_{3}[2] - \frac{1}{4} \mathbf{g}_{7}[2] - \frac{1}{6} \mathbf{g}_{9}[2], \\
        \Tilde{\mathbf{g}}_{(1,2)} &\defeq \frac{4}{9} \mathbf{g}_{1}[1] + \frac{1}{3} \mathbf{g}_{3}[1]  - \frac{4}{3} \mathbf{g}_{4}[1] - \mathbf{g}_{6}[1] \nonumber \\
        &-\frac{1}{6} \mathbf{g}_{1}[2] - \frac{1}{9} \mathbf{g}_{3}[2] + \frac{1}{2} \mathbf{g}_{4}[2] + \frac{1}{3} \mathbf{g}_{6}[2], \\
        \Tilde{\mathbf{g}}_{(1,3)} &\defeq -\frac{8}{3} \mathbf{g}_{4}[1] - 2 \mathbf{g}_{6}[1]  - \frac{2}{3} \mathbf{g}_{7}[1] - \frac{1}{2}\mathbf{g}_{9}[1] \nonumber \\
        &+ \mathbf{g}_{4}[2] + \frac{2}{3} \mathbf{g}_{6}[2] + \frac{1}{4} \mathbf{g}_{7}[2] + \frac{1}{6} \mathbf{g}_{9}[2].
    \end{align*}
    Note that each worker in the cluster computes a vector of size $\frac{d}{2}$. We assume that $W_{1,3}$ is a straggler.
    After receiving computations from $W_{1,1}$ and $W_{1,2}$, Relay 1 interpolates the following polynomial function, referred to as Intra-Cluster  encoding polynomial for Cluster 1, 
    \begin{align*}
        \mathbf{f}_{\text{intraC}}^{(1)}(x_2)&\defeq
         \mathbf{g}_{1}[1] \frac{4(x_2-3)}{-9} + \mathbf{g}_{3}[1] \frac{x_2-3}{-3} + \mathbf{g}_{4}[1] \frac{4(x_2-1)}{-3} \\
        &+ \mathbf{g}_{6}[1] \frac{x_2-1}{-1} + \mathbf{g}_{7}[1] \frac{2(x_2-2)}{-3} + \mathbf{g}_{9}[1] \frac{x_2-2}{-2} \nonumber \\
        &+ \mathbf{g}_{1}[2] \frac{x_2-3}{6} + \mathbf{g}_{3}[2] \frac{x_2-3}{9} + \mathbf{g}_{4}[2] \frac{x_2-1}{2} \\
        &+ \mathbf{g}_{6}[2] \frac{x_2-1}{3} + \mathbf{g}_{7}[2] \frac{x_2-2}{4} + \mathbf{g}_{9}[2] \frac{x_2-2}{6}.
    \end{align*}
    This is possible since the polynomial is of degree 1 and can be recovered by any 2 evaluations and that $\Tilde{\mathbf{g}}_{(1,j)}$ computed by $W_{1,j}$ is equal to $\mathbf{f}_{\text{intraC}}^{(1)}(x_2=j)$. Then, Relay 1 computes $\tilde{\mathbf{g}}_1\defeq\mathbf{f}_{\text{intraC}}^{(1)}(x_2=0)$ and sends the result, which is a vector of size $\frac{d}{2}$, to the server. 
    For Cluster $2$ and Cluster $3$, the corresponding Intra-Cluster polynomials are defined as follows.
    \begin{align*}
        \mathbf{f}_{\text{intraC}}^{(2)}(x_2)&\defeq
        \mathbf{g}_{1}[1] \frac{x_2-3}{-3} + \mathbf{g}_{2}[1] \frac{x_2-3}{1} + \mathbf{g}_{4}[1] \frac{x_2-1}{-1} + \mathbf{g}_{5}[1] \frac{3(x_2-1)}{1} + \mathbf{g}_{7}[1] \frac{x_2-2}{-2} + \mathbf{g}_{8}[1] \frac{3(x_2-2)}{2} \nonumber \\
        &+ \mathbf{g}_{1}[2] \frac{x_2-3}{6} + \mathbf{g}_{2}[2] \frac{x_2-3}{-3} + \mathbf{g}_{4}[2] \frac{x_2-1}{2} \nonumber\\
        &+ \mathbf{g}_{5}[2] \frac{x_2-1}{-1} + \mathbf{g}_{7}[2] \frac{x_2-2}{4} + \mathbf{g}_{8}[2] \frac{x_2-2}{-2}, \\
        \mathbf{f}_{\text{intraC}}^{(3)}(x_2)&\hspace{-1mm}\defeq\hspace{-1mm}
        \mathbf{g}_{2}[1] \frac{8(x_2-3)}{3} + \mathbf{g}_{3}[1] \frac{2(x_2-3)}{3} + \mathbf{g}_{5}[1] \frac{8(x_2-1)}{1} + \mathbf{g}_{6}[1] \frac{2(x_2-1)}{1} + \mathbf{g}_{8}[1] \frac{4(x_2-2)}{1} + \mathbf{g}_{9}[1] \frac{x_2-2}{1} \nonumber \\
        &+ \mathbf{g}_{2}[2] \frac{x_2-3}{-1} + \mathbf{g}_{3}[2] \frac{x_2-3}{-3} + \mathbf{g}_{5}[2] \frac{3(x_2-1)}{1} \nonumber\\
        &+ \mathbf{g}_{6}[2] \frac{x_2-1}{-1} + \mathbf{g}_{8}[2] \frac{3(x_2-2)}{-2} + \mathbf{g}_{9}[2] \frac{x_2-2}{-2}.
    \end{align*}
    Similarly, each non-straggling $W_{n,j}$, for $n=2,3$ and $j=1,2,3$, sends the value of $\mathbf{f}_{\text{intraC}}^{(n)}(x_2=j)$ to its parent relay. Relay $n$ tolerating one straggler, is then able to interpolate $\mathbf{f}_{\text{intraC}}^{(n)}(x_2)$.
    Next, Relay $n$ computes $\Tilde{\mathbf{g}}_{n} \defeq\mathbf{f}_{\text{intraC}}^{(n)}(x_2=0)$ and  sends the result to the server, which is as follows. 
    \begin{align*}
        \Tilde{\mathbf{g}}_{1}&= 
        \frac{4}{3} \mathbf{g}_{1}[1]+\mathbf{g}_{3}[1]+\frac{4}{3} \mathbf{g}_{4}[1]+\mathbf{g}_{6}[1]+\frac{4}{3} \mathbf{g}_{7}[1]+\mathbf{g}_{9}[1] \nonumber \\
        & -\frac{1}{2} \mathbf{g}_{1}[2]-\frac{1}{3} \mathbf{g}_{3}[2]-\frac{1}{2} \mathbf{g}_{4}[2]-\frac{1}{3} \mathbf{g}_{6}[2]-\frac{1}{2} \mathbf{g}_{7}[2]-\frac{1}{3} \mathbf{g}_{9}[2], \\
              \Tilde{\mathbf{g}}_{2} &=
        \mathbf{g}_{1}[1]-3\mathbf{g}_{2}[1]+\mathbf{g}_{4}[1]-3\mathbf{g}_{5}[1]+\mathbf{g}_{7}[1]-3\mathbf{g}_{8}[1] \nonumber \\
        & -\frac{1}{2} \mathbf{g}_{1}[2]+\mathbf{g}_{2}[2]-\frac{1}{2}\mathbf{g}_{4}[2]+\mathbf{g}_{5}[2]-\frac{1}{2}\mathbf{g}_{7}[2]+\mathbf{g}_{8}[2], \\
                \Tilde{\mathbf{g}}_{3}& =
         -8\mathbf{g}_{2}[1]-2\mathbf{g}_{3}[1]-8\mathbf{g}_{5}[1]-2\mathbf{g}_{6}[1]-8\mathbf{g}_{8}[1]-2\mathbf{g}_{9}[1] \nonumber \\
        & +3\mathbf{g}_{2}[2]+\mathbf{g}_{3}[2]+3\mathbf{g}_{5}[2]+\mathbf{g}_{6}[2]+3\mathbf{g}_{8}[2]+\mathbf{g}_{9}[2]. 
      \end{align*}
  
    It is important to highlight the linear structure of the coding, where the coefficients have been determined through designing polynomials. For instance, one can verify that $\Tilde{\mathbf{g}}_{1} = 2\Tilde{\mathbf{g}}_{(1,1)} - \Tilde{\mathbf{g}}_{(1,2)}$. Ultimately, the computation at each relay is a linear combination of the computations of the non-straggling workers connected to it. The coefficients of the worker computations are determined by the Intra-Cluster polynomials.
    
    After receiving these computations, the server recovers the following polynomial, referred to as the Cluster-to-Server encoding polynomial.
    \begin{align*}
        &\mathbf{f}_{\text{C2S}}(x_1)= 
         \left( \mathbf{g}_{1}[1] \frac{x_1-3}{-3} + \mathbf{g}_{2}[1] \frac{x_1-1}{-1} + \mathbf{g}_{3}[1] \frac{x_1-2}{-2} \right. \nonumber \\ 
        &+ \left. \mathbf{g}_{4}[1] \frac{x_1-3}{-3} + \mathbf{g}_{5}[1] \frac{x_1-1}{-1} + \mathbf{g}_{6}[1] \frac{x_1-2}{-2} \right. + \left. \mathbf{g}_{7}[1] \frac{x_1-3}{-3} + \mathbf{g}_{8}[1] \frac{x_1-1}{-1} + \mathbf{g}_{9}[1] \frac{x_1-2}{-2} \right) (x_1+1) \nonumber \\
        &+ \left( \mathbf{g}_{1}[2] \frac{x_1-3}{-4} + \mathbf{g}_{2}[2] \frac{x_1-1}{-2} + \mathbf{g}_{3}[2] \frac{x_1-2}{-3} + \mathbf{g}_{4}[2] \frac{x_1-3}{-4} \right. \nonumber \\ 
        &+ \left. \mathbf{g}_{5}[2] \frac{x_1-1}{-2} + \mathbf{g}_{6}[2] \frac{x_1-2}{-3} + \mathbf{g}_{7}[2] \frac{x_1-3}{-4} + \mathbf{g}_{8}[2] \frac{x_1-1}{-2} \right. + \left. \mathbf{g}_{9}[2] \frac{x_1-2}{-3} \right) (-x_1).
    \end{align*}
    This is possible since $\Tilde{\mathbf{g}}_{n}$, sent by Relay $n$, is equal to the evaluation of $\mathbf{f}_{\text{C2S}}(x_1)$ at the point $x_1=n$. This polynomial is of degree 2 and can be recovered by 3 evaluations. Since in this example, there are no stragglers among the relays, the server can recover this polynomial upon receiving     
    $\Tilde{\mathbf{g}}_{1}$, $\Tilde{\mathbf{g}}_{2}$, and $\Tilde{\mathbf{g}}_{3}$.
    
    Besides, One can verify that $\mathbf{f}_{\text{C2S}}(x_1=0)$ is equal to $\mathbf{g}[1]=\sum_{k=1}^{9}\mathbf{g}_k[1]$, and  $\mathbf{f}_{\text{C2S}}(x_1=-1)$ is equal to $\mathbf{g}[2]=\sum_{k=1}^{9}\mathbf{g}_k[2]$, thereby recovering $\mathbf{g}=(\mathbf{g}[1], \mathbf{g}[2])$. To highlight the linear structure of the messages, one can verify $\mathbf{f}_{\text{C2S}}(x_1=0) = 3\Tilde{\mathbf{g}}_{1} -3\Tilde{\mathbf{g}}_{2} +\Tilde{\mathbf{g}}_{3}$ and $\mathbf{f}_{\text{C2S}}(x_1=-1) = 6\Tilde{\mathbf{g}}_{1} -8\Tilde{\mathbf{g}}_{2} +3\Tilde{\mathbf{g}}_{3}$. The coefficients of the relay computations are determined by the Cluster-to-Server polynomial so that the final gradient computation could be recovered.

Note that since each non-straggling worker and relay sends a linear combination of the $\mathbf{g}_k[i]$s, each of size  $\frac{d}{2}$, the communication costs are $C_1=C_2^{(n)}=\frac{d}{2}, \forall n \in [N_1]$. 

    \subsection{The Proposed Hierarchical Gradient Coding Scheme} \label{sec:generalscheme}
 In this section, we formally describe the proposed gradient coding scheme in a hierarchical setting, which achieves optimal communication load under a heterogeneous data assignment among workers; the number of datasets assigned to each worker is not necessarily the same and can be different based on the computational power of each worker. Building on the universal polynomial function proposed in \cite{jahani2021optimal}, we extend this idea to design two encoding polynomials, referred to as the \emph{Intra-Cluster} and \emph{Cluster-to-Server} encoding polynomials. These polynomials are specifically tailored to the hierarchical setting, enabling efficient communication and computation in this structure.
    
Each partial gradient, which is a vector of size $d$, is partitioned into $m_1$ parts as $\mathbf{g}_k=(\mathbf{g}_k[1], \dots, \mathbf{g}_k[m_1])$, used to design the cluster-to-server polynomial. Furthermore, to design the intra-cluster polynomial in Cluster $n$, each part is further divided into $m_2^{(n)}$ subparts, represented as $\mathbf{g}_k[i]=(\mathbf{g}_k[i,1], \dots, \mathbf{g}_k[i,m_2^{(n)}])$. Consequently, each partial gradient can be expressed by $\mathbf{g}_k=(\mathbf{g}_k[1,1], ..., \mathbf{g}_k[m_1,m_2^{(n)}])$, where $m_1 \defeq r_1-2a_1-s_1$ and $m_2^{(n)} \defeq r_{2,n}-2a_2^{(n)}-s_2^{(n)}$. If $m_1m_2^{(n)} \nmid d$, we can zero pad the partial gradient.
   

    \subsubsection{Design of the Encoding Polynomials}
The encoding process is designed in two layers: the Cluster-to-Server part and the Intra-Cluster part. The first layer specifies the computations to be carried out by the relays, while the second layer defines the computations to be performed by the workers. These two polynomials are interconnected and are designed based on the system parameters and the data assignment.
We begin by describing the design of the Cluster-to-Server encoding polynomial below.\\    
\noindent\textbf{Design of the Cluster-to-Server Encoding: }  Let $\alpha_1^{(1)}, ..., \alpha_{N_1}^{(1)}$ and $\beta_1^{(1)}, ..., \beta_{m_1}^{(1)}$ be $N_1+m_1$ distinct values  chosen from $\mathbb{F}_q$. We aim to find a polynomial function, denoted by $\mathbf{f}_{\text{C2S}}(x_1)$, which has the following properties:
    \begin{enumerate}[leftmargin=*, noitemsep]
        \item  $\mathbf{f}_{\text{C2S}}(\alpha_n^{(1)})$ is a function of the datasets in $\Gamma_n$.  Here, $\alpha_n^{(1)}$ represents the point where $\mathbf{f}_{\text{C2S}}(x_1)$ is evaluated by Relay $n$. Since Relay $n$ has access only to the datasets in $\Gamma_n$, its computations are limited to those datasets. During the scheme, the computed evaluation by Relay $n$ is sent to the server.
        \item  $\mathbf{f}_{\text{C2S}}(x_1)$ is a polynomial of degree $N_1-2a_1-s_1-1$. This ensures that the server can recover the polynomial using $N_1-2a_1-s_1$ evaluations, even in the presence of $s_1$ straggler relays whose evaluations are not received and $a_1$ adversarial relays that may provide incorrect evaluations. This is analogous to a Reed–Solomon code of length $N_1$ and dimension $N_1 - 2a_1 - s_1$, which can jointly correct up to $a_1$ errors and $s_1$ erasures.\cite{lin2001error}.
        \item At point $\beta_i^{(1)}$, $\mathbf{f}_{\text{C2S}}(x_1)$ satisfies the following condition:
        \begin{align*}
           \mathbf{f}_{\text{C2S}}(\beta_i^{(1)}) = \sum_{k=1}^K \mathbf{g}_k[i] = \mathbf{g}[i].
        \end{align*}
      When the server recovers $\mathbf{f}_{\text{C2S}}(x_1)$, it computes the  $i$-th  element of the desired gradient by evaluating the polynomial at $\beta_i^{(1)}$, thereby recovering the entire gradient vector.
    \end{enumerate}
 


  In order to design $\mathbf{f}_{\text{C2S}}(x_1)$,  we first define the polynomial $\mathbf{p}_{l_1}(x_1)$ as follows,
    \begin{align}
        \mathbf{p}_{l_1}(x_1)=\sum_{k=1}^{K} \mathbf{g}_{k}[l_1] \prod^{N_1}_{\substack{i=1 \\ i:\mathcal{D}_{k}\notin\Gamma_{i}}} \frac{x_1-\alpha^{(1)}_{i}}{\beta^{(1)}_{l_1}-\alpha^{(1)}_{i}}.
    \end{align}
    The property of this polynomial is that $\mathbf{p}_{l_1}(\beta_{l_1}^{(1)})=\sum_{k=1}^K \mathbf{g}_k[l_1]$, and  $\mathbf{p}_{l_1}(\alpha_{n}^{(1)})$ depends only on the partial gradients of datasets in $\Gamma_n$.
   The polynomial $\mathbf{f}_{\text{C2S}}(x_1)$ with the mentioned properties is created as follows.
    \begin{align} \label{eq:outer}
        \mathbf{f}_{\text{C2S}}(x_1)=\sum_{l_1=1}^{m_1} \bigg(\mathbf{p}_{l_1}(x_1)  \prod^{m_1}_{\substack{u_1=1 \\ u_1 \neq l_1}} \frac{x_{1}-\beta^{(1)}_{u_1}}{\beta^{(1)}_{l_1}-\beta^{(1)}_{u_1}}\bigg).
    \end{align}
    
    In the proposed scheme, Relay $n$ sends the evaluation of  $\mathbf{f}_{\text{C2S}}(x_1)$ at point $x_1=\alpha_n^{(1)}$ to the server. More precisely, 
    \begin{align*} 
        \Tilde{\mathbf{g}}_n\defeq
        \sum_{l_1=1}^{m_1} \left(\sum_{k=1}^{K} \mathbf{g}_{k}[l_1] \prod^{N_1}_{\substack{i=1 \\ i:\mathcal{D}_{k}\notin\Gamma_{i}}} \frac{\alpha_n^{(1)}-\alpha^{(1)}_{i}}{\beta^{(1)}_{l_1}-\alpha^{(1)}_{i}} \right)  \prod^{m_1}_{\substack{u_1=1 \\ u_1 \neq l_1}} \frac{\alpha_n^{(1)}-\beta^{(1)}_{u_1}}{\beta^{(1)}_{l_1}-\beta^{(1)}_{u_1}},
    \end{align*}
    is sent to the server by Relay $n$ where $\mathbf{g}_{k}[l_1] = (\mathbf{g}_{k}[l_1,1],...,\mathbf{g}_{k}[l_1,m_2^{(n)}])$, thereby $\Tilde{\mathbf{g}}_n = (\Tilde{\mathbf{g}}_n[1],...,\Tilde{\mathbf{g}}_n[m_2^{(n)}])$. Since this is a linear combination of $\mathbf{g}_{k}[l_1]$s of size $\frac{d}{m_1}$, the relay to-server communication cost would be $C_1=\frac{d}{m_1}$. Note that $\Tilde{\mathbf{g}}_n$ is a linear combination of the partitions of the partial gradients of datasets in $\Gamma_n$.

 \noindent   \textbf{Design of the Intra-Cluster Encoding: } For each Cluster $n$, we need a separate Intra-Cluster encoding polynomial accustomed to the corresponding Cluster-to-Server encoding and $\Tilde{\mathbf{g}}_n$. Let $\alpha_1^{(2)}, ..., \alpha_{N_2^{(n)}}^{(2)}$ and $\beta_1^{(2)}, ..., \beta_{m_2^{(n)}}^{(2)}$ be $N_2^{(n)}+m_2^{(n)}$ distinct values chosen from $\mathbb{F}_q$. We aim to find a polynomial function, denoted by  $\mathbf{f}_{\text{intraC}}^{(n)}(x_2)$, which has the following properties:
\begin{enumerate}[leftmargin=*, noitemsep]
    \item  $\mathbf{f}_{\text{intraC}}^{(n)}(\alpha_j^{(2)})$ is a function of datasets in $\Gamma_{(n,j)}$.  Specifically, $\alpha_j^{(2)}$ represents the point at which $\mathbf{f}_{\text{intraC}}^{(n)}(x_2)$ is evaluated by $W_{(n,j)}$ and then sent to Relay $n$. $W_{(n,j)}$ can only perform computations on the datasets within $\Gamma_{(n,j)}$. 
    \item  $\mathbf{f}_{\text{intraC}}^{(n)}(x_2)$ is a polynomial of degree $N_2^{(n)}-2a_2^{(n)}-s_2^{(n)}-1$.  This degree is chosen to ensure that the polynomial can be recovered from the evaluations of the workers in the cluster. Since $s_2^{(n)}$ workers are stragglers whose evaluations are not received by the relay, and $a_2^{(n)}$ workers may behave adversarially by sending incorrect evaluations, it is necessary for the relay to recover the polynomial using  $N_2^{(n)}-2a_2^{(n)}-s_2^{(n)}$ evaluations. 
    \item  At point $\beta_j^{(2)}$, $\mathbf{f}_{\text{intraC}}^{(n)}(x_2)$ satisfies the following condition:
    \begin{align*}
     \mathbf{f}_{\text{intraC}}^{(n)}(\beta_j^{(2)}) = \tilde{\mathbf{g}}_n[j].   
    \end{align*}
   In other words, after Relay $n$ recovers
 $\mathbf{f}_{\text{intraC}}^{(n)}(x_2)$, it calculates the $j$-th element of $\tilde{\mathbf{g}}_n$  by evaluating the polynomial at $\beta_j^{(2)}$, thereby computing the entire vector $\tilde{\mathbf{g}}_n$.
\end{enumerate}



    We first define the polynomial $\mathbf{p}_{l_2}^{(n)}(x_2)$ as follows.
    \begin{align} \label{eq:p}
        &\mathbf{p}_{l_2}^{(n)}(x_2)= \sum_{l_1=1}^{m_1}\hspace{-1mm}\left(\sum_{k=1}^{K} \mathbf{g}_{k}[l_1,l_2]\prod^{N_1}_{\substack{i=1 \\ i:\mathcal{D}_{k}\notin\Gamma_{i}}} \frac{\alpha_n^{(1)}-\alpha^{(1)}_{i}}{\beta^{(1)}_{l_1}-\alpha^{(1)}_{i}}
        \prod^{N_2^{(n)}}_{\substack{j=1 \\ j:\mathcal{D}_{k}\notin\Gamma_{(n,j)}}} \hspace{-2mm}\frac{x_2-\alpha^{(2)}_{j}}{\beta^{(2)}_{l_2}-\alpha^{(2)}_{j}} \right) \prod^{m_1}_{\substack{u_1=1 \\ u_1 \neq l_1}} \frac{\alpha_1^{(1)}-\beta^{(1)}_{u_1}}{\beta^{(1)}_{l_1}-\beta^{(1)}_{u_1}}.
    \end{align}
    The property of this polynomial is that $\mathbf{p}_{l_2}^{(n)}(\beta_{l_2}^{(2)})= \Tilde{\mathbf{g}}_n[l_2]$, and $\mathbf{p}_{l_2}^{(n)}(\alpha_{j}^{(2)})$ depends only on the partial gradients of datasets in $\Gamma_{(n,j)}$. The Intra-Cluster encoding polynomial for Cluster $n$, denoted by $ \mathbf{f}_{\text{intraC}}^{(n)}(x_2)$, is defined as follows. 
    \begin{align}
        \mathbf{f}_{\text{intraC}}^{(n)}(x_2) =
        \sum_{l_2=1}^{m_2^{(n)}} \mathbf{p}_{l_2}^{(n)}(x_2) \prod^{m_2^{(n)}}_{\substack{u_2=1 \\ u_2 \neq l_2}} \frac{x_{2}-\beta^{(2)}_{u_2}}{\beta^{(2)}_{l_2}-\beta^{(2)}_{u_2}}.
    \end{align}
    In the proposed scheme, the task of $W_{(n,j)}$ is to calculate 
    $\mathbf{f}_{\text{intraC}}^{(n)}(\alpha_j^{(2)})$, i.e.,
    \begin{align} \label{eq:workertask}
        \Tilde{\mathbf{g}}_{n,j} \defeq \sum_{l_2=1}^{m_2^{(n)}} \mathbf{p}_{l_2}^{(n)}(\alpha_j^{(2)}) \prod^{m_2^{(n)}}_{\substack{u_2=1 \\ u_2 \neq l_2}} \frac{\alpha_j^{(2)}-\beta^{(2)}_{u_2}}{\beta^{(2)}_{l_2}-\beta^{(2)}_{u_2}},
    \end{align}
and send the result to Relay $n$. Since this is a linear combination of $\mathbf{g}_{k}[l_1,l_2]$s of size $\frac{d}{m_1m_2^{(n)}}$, the worker-to-relay communication cost would be $C_2^{(n)}=\frac{d}{m_1m_2^{(n)}}$.
\subsubsection{Hierarchical Gradient Coding: Step-by-Step Algorithm}
   With the provided description, the proposed scheme proceeds as follows:
\begin{itemize}
    \item \textbf{Worker-to-Relay Transmission.} $W_{(n,j)}$ computes  $\Tilde{\mathbf{g}}_{(n,j)}$ as defined in \eqref{eq:workertask}, and sends the result to Relay $n$. 
    \item \textbf{Calculating Intra-Cluster Polynomial by the Relay.} Upon receiving $\Tilde{\mathbf{g}}_{(n,j)}$ from non-straggling workers, Relay $n$ recovers $\mathbf{f}_{\text{intraC}}^{(n)}(x_2)$, which is of degree $N_2^{(n)}-2a_2^{(n)}-s_2^{(n)}-1$. 
    \item \textbf{Relay-to-Server Transmission.} Relay $n$,  utilizing the third property of the Intra-Cluster polynomial, computes $\Tilde{\mathbf{g}}_{n} = (\Tilde{\mathbf{g}}_{n}[1],...,\Tilde{\mathbf{g}}_{n}[m_2^{(n)}])$ and sends it to the server.
    \item  \textbf{Calculating Cluster-to-Server Polynomial by the Server.} The server, having received $\Tilde{\mathbf{g}}_{n}$ from non-straggling relays, recovers $\mathbf{f}_{\text{C2S}}(x_1)$,  which is of degree  $N_1-2a_1-s_1-1$. 
    \item \textbf{Final Gradient Recovery.}  Using the third property of the Cluster-to-Server polynomial, the server calculates the aggregated partial gradients ${\mathbf{g}} = ({\mathbf{g}}[1],...,{\mathbf{g}}[m_1])$.
    
\end{itemize}
  

    \begin{remark}
In the proposed scheme, the optimal communication loads are achieved simultaneously through the use of polynomial coding techniques in both ``the workers and the relays". In fact, the relays do not only forward the received vectors; instead, they serve as intermediate nodes that perform computations on the information received from the workers.
    \end{remark}

    \begin{remark}
        Let us compare the hierarchical setting to a non-hierarchical setting with the same parameters. For simplicity of comparison, we assume the parameters of each cluster in the hierarchical setting is the same; all having $s_2$ stragglers and $N_2$ workers. This would lead to a total of $N_1s_2$ straggling workers. Using the result in \cite{jahani2021optimal}, the total communication load received by the server from all workers in the non-hierarchical scheme with the same parameters would be $\frac{N_1N_2d}{r_T-N_1s_2}$, where $r_T$ denotes the minimum number of copies of each dataset among workers. However, in the proposed hierarchical scheme, the load received by the server (considering no straggling and no adversarial relays) would be $\frac{N_1d}{r_1}$. Since $r_T$ would be at most $r_1N_2$,  it follows that the total load in the hierarchical setting is less than that in the non-hierarchical setting, highlighting the advantage of the hierarchical structure with respect to bandwidth limitations at the server. 
    \end{remark}
     \begin{remark}
        The result presented in \cite{tang2024design} is a special case of our proposed scheme, specifically when  $r_1=s_e+1$ and $r_{2,n}=s_w+1$ for each Cluster $n$, without considering adversarial nodes. Here $s_e$ and $s_w$
  represent the number of stragglers among relays and in each cluster, respectively, as defined
  in \cite{tang2024design}. Note that, in this case, both the relay to server and the worker to relay communication loads will be  $C_1=C_{2}^{(n)}=d, \forall n \in [N_1]$, matching the result in \cite{tang2024design}.
    \end{remark}

\section{Privacy-Preserving Hierarchical Gradient Coding} \label{sec:settingprivacy}

In this section, we present our proposed privacy-preserving hierarchical gradient coding scheme based on privacy defined in Subsection \ref{sec:problemsettingwithprivacy}, which ensures that the computed partial gradients remain private from the relays. We first outline the main results in Subsection~\ref{sec:resultsprivacy}. The scheme is introduced through an example in Subsection~\ref{sec:exmpprivacy}, followed by a detailed description of the general scheme in Subsection~\ref{sec:achprivacy}. Finally, the privacy proof is provided in Subsection~\ref{sec:proofprivacy}.


\subsection{Main Results} \label{sec:resultsprivacy}
    \begin{theorem}
        For the privacy-preserving hierarchical gradient coding problem, as defined in Subsection \ref{sec:problemsettingwithprivacy}, and for values $r_1 \in [s_1+1:N_1-1]$, $r_{2,n}=s_2^{(n)}+1, \forall n\in [N_1]$, and $a_2^{(n)}=0, \forall n\in [N_1]$, the minimum communication loads are characterized by
        \begin{align}
            {C_1^{\pi}}^* = {C_2^{\pi^{(n)}}}^{*} = \frac{d}{m_1},
        \end{align}
        where $m_1=r_1-2a_1-s_1$.
    \end{theorem}

\begin{proof}
    The optimality follows directly from Theorem \ref{thm:main}, since in that Theorem $m^{(n)}_2=r_{2,n}-2a^{(n)}_2-s^{(n)}_2$, which would be equal to $1$ for all values of $n \in [N_1]$, under the system parameters considered here. Since we have an additional privacy constraint compared to the previous setting, but the loads have not increased for the given parameters. The achievability proof is provided in Subsection \ref{sec:achprivacy} and the privacy guarantee is proven in Subsection \ref{sec:proofprivacy}.
\end{proof}
\subsection{Motivating Example} \label{sec:exmpprivacy}
   In this subsection, we explain the proposed scheme through an example. Fig.~\ref{fig:diagram exmp2} illustrates a hierarchical setting consisting of $N_1=4$ clusters, each containing $3$ workers, i.e., $N_2^{(1)}=N_2^{(2)}=N_2^{(3)}=3$. The dataset is partitioned into $K=4$ equal-sized datasets, and the data assignment is performed such that  $r_1=2$, and  $r_{2,i}=2$ for $i\in[4]$. 

\begin{figure}
\centering
\resizebox{0.7\columnwidth}{!}{%
\begin{tikzpicture}[every node/.style={font=\small}]

\def\xspacing{2}

\foreach \i [evaluate=\x as \x using int((\i-1)*\xspacing)] in {1,2,3,4} {
    \node[circle, draw, minimum size=0.3cm, inner sep=1pt, font=\tiny, scale=0.7, fill=black!5] (W\i1) at (\x-0.6, 0) {$W_{\i,1}$};
    \node[circle, draw, minimum size=0.3cm, inner sep=1pt, font=\tiny, scale=0.7, fill=black!5] (W\i2) at (\x, 0) {$W_{\i,2}$};
    \node[circle, draw, minimum size=0.3cm, inner sep=1pt, font=\tiny, scale=0.7, fill=black!5] (W\i3) at (\x+0.6, 0) {$W_{\i,3}$};

    \node[draw, rectangle, minimum width=1.4cm, minimum height=0.6cm, fill=black!5] (R\i) at (\x, -1.2) {Relay \i};

    \foreach \j in {1,2,3} {
        \draw (W\i\j) -- (R\i);
    }

    \node[draw, dashed, fit={(W\i1) (W\i2) (W\i3) (R\i)}, label=above:Cluster \i] {};
}

\coordinate (ServerCoord) at ($ (R1)!0.5!(R4) + (0,-2) $);

\node[draw, rectangle, minimum width=1.4cm, minimum height=0.6cm, fill=black!5] (Server) at (ServerCoord) {Server};

\foreach \i in {1,2,3,4} {
    \draw (R\i) -- (Server);
}

\end{tikzpicture}
}
\caption{Structure of the motivating example.}
\label{fig:diagram exmp2}
\end{figure}

    Consider the dataset assignment is performed as follows.
    \begin{align*} 
        \Gamma_{(1,1)} &= \{\mathcal{D}_1, \mathcal{D}_3\},
        \Gamma_{(1,2)} = \{\mathcal{D}_1\},
        \Gamma_{(1,3)} = \{\mathcal{D}_3\}, \\
        \Gamma_{(2,1)} &= \{\mathcal{D}_1, \mathcal{D}_4\},
        \Gamma_{(2,2)} = \{\mathcal{D}_1\},
        \Gamma_{(2,3)} = \{\mathcal{D}_4\}, \\
        \Gamma_{(3,1)} &= \{\mathcal{D}_2, \mathcal{D}_3\},
        \Gamma_{(3,2)} = \{\mathcal{D}_2\},
        \Gamma_{(3,3)} = \{\mathcal{D}_3\}, \\
        \Gamma_{(4,1)} &= \{\mathcal{D}_2, \mathcal{D}_4\},
        \Gamma_{(4,2)} = \{\mathcal{D}_2\},
        \Gamma_{(4,3)} = \{\mathcal{D}_4\}, \\
    \end{align*}
    which leads that
    \begin{align*}
        & \Gamma_{1} = \{\mathcal{D}_1, \mathcal{D}_3\}, \\
        & \Gamma_{2} = \{\mathcal{D}_1, \mathcal{D}_4\}, \\
        & \Gamma_{3} = \{\mathcal{D}_2, \mathcal{D}_3\}, \\
        & \Gamma_{4} = \{\mathcal{D}_2, \mathcal{D}_4\}.
    \end{align*}

    Besides the assignment of datasets,  the three random vectors assigned to the workers, as defined in Subsection \ref{sec:problemsettingwithprivacy}, are as follows.
    \begin{align*} 
        \Gamma'_{(1,1)} &= \{\mathbf{z}_2, \mathbf{z}_3\},
        \Gamma'_{(1,2)} = \{\mathbf{z}_1, \mathbf{z}_2, \mathbf{z}_3\},
        \Gamma'_{(1,3)} = \{\mathbf{z}_1, \mathbf{z}_3\}, \\
        \Gamma'_{(2,1)} &= \{\mathbf{z}_1, \mathbf{z}_2\},
        \Gamma'_{(2,2)} = \{\mathbf{z}_1, \mathbf{z}_2\},
        \Gamma'_{(2,3)} = \{\mathbf{z}_2\}, \\
        \Gamma'_{(3,1)} &= \{\mathbf{z}_1, \mathbf{z}_3\},
        \Gamma'_{(3,2)} = \{\mathbf{z}_1, \mathbf{z}_3\},
        \Gamma'_{(3,3)} = \{\mathbf{z}_3\}, \\
        \Gamma'_{(4,1)} &= \{\mathbf{z}_2, \mathbf{z}_3\},
        \Gamma'_{(4,2)} = \{\mathbf{z}_2, \mathbf{z}_3\},
        \Gamma'_{(4,3)} = \{ \mathbf{z}_3\}.
    \end{align*}
    This leads to the following cluster-wise assignments:
    \begin{align*}
        & \Gamma'_{1} = \{\mathbf{z}_1, \mathbf{z}_2, \mathbf{z}_3\}, \\
        & \Gamma'_{2} = \{\mathbf{z}_1, \mathbf{z}_2\}, \\
        & \Gamma'_{3} = \{\mathbf{z}_1, \mathbf{z}_3\}, \\
        & \Gamma'_{4} = \{\mathbf{z}_2, \mathbf{z}_3\}.
    \end{align*}

    We assume there are no adversaries, but a single straggler among relays and in each cluster; i.e., $a_1=0, s_1=1$ and $s_2^{(i)}=1$ for $i\in[4]$. In addition, note that in the privacy-preserving scenario, no adversaries inside clusters are allowed, i.e., $a_2^{(i)}=0$ for $i\in[4]$. This gives $m_1=1$ and $m_2^{(i)}=1$ for $i\in[4]$, meaning the partial gradients and random vectors will not be partitioned. 
    Our focus in this example is to demonstrate how privacy is achieved.
    

    The workers in Cluster 1 compute $\Tilde{\mathbf{g}}_{(1,1)}^{\pi}$, $\Tilde{\mathbf{g}}_{(1,2)}^{\pi}$ and $\Tilde{\mathbf{g}}_{(1,3)}^{\pi}$ as follows,
    \begin{align*}
        \Tilde{\mathbf{g}}_{(1,1)}^{\pi} &\defeq \frac{1}{3} \mathbf{g}_{1} + \frac{3}{16} \mathbf{g}_{3} + \frac{4}{9} \mathbf{z}_{2} + \frac{1}{2} \mathbf{z}_{3}, \\
        \Tilde{\mathbf{g}}_{(1,2)}^{\pi} &\defeq \frac{1}{6} \mathbf{g}_{1} - \frac{3}{4} \mathbf{z}_{1} + \frac{2}{9} \mathbf{z}_{2} + \frac{1}{2}\mathbf{z}_{3}, \\
        \Tilde{\mathbf{g}}_{(1,3)}^{\pi} &\defeq -\frac{3}{16} \mathbf{g}_{3} - \frac{3}{2} \mathbf{z}_{1} + \frac{1}{2}\mathbf{z}_{3}.
    \end{align*}
    Assuming $W_{(1,3)}$ is a straggler, Relay 1 receives two computations from $W_{(1,1)}$ and $W_{(1,2)}$. With these two computations, Relay 1 can interpolate the following polynomial, referred to as the Intra-Cluster polynomial of Cluster 1. 
    \begin{align*}
        {\mathbf{f}_{\text{intraC}}^{\pi}}^{(1)}(x_2)&\defeq
         \mathbf{g}_{1} \frac{x_2-3}{-6} + \mathbf{g}_{3} \frac{3(x_2-2)}{-16} + \mathbf{z}_{1} \frac{3(x_2-1)}{-4} + \mathbf{z}_{2} \frac{2(x_2-3)}{-9} + \mathbf{z}_{3} \frac{1}{2}.
    \end{align*}
    This is possible since ${\mathbf{f}_{\text{intraC}}^{\pi}}^{(1)}(x_2)$ is of degree 1, and the computations of the $W_{(1,1)}$ and $W_{(1,2)}$ are in fact ${\mathbf{f}_{\text{intraC}}^{\pi}}^{(1)}(\alpha_1^{(1)})$ and ${\mathbf{f}_{\text{intraC}}^{\pi}}^{(1)}(\alpha_2^{(1)})$ respectively, where $\alpha_1^{(1)}=1$ and $\alpha_2^{(1)}=2$.

    The reason Relay 1 cannot obtain any information about the partial gradients is as follows. Consider the random vector component of the two received computations, which can be expressed as:
    \begin{align*}
        \begin{bmatrix}
        0 & \frac{4}{9} & \frac{1}{2} \\
        -\frac{3}{4} & \frac{2}{9} & \frac{1}{2}
    \end{bmatrix}
    \begin{bmatrix}
        \mathbf{z}_{1} \\
        \mathbf{z}_{2} \\
        \mathbf{z}_{3}
    \end{bmatrix},
    \end{align*}
    where each row corresponds to a worker. Since the coefficient matrix has rank 2, there are no non-trivial linear combinations of the rows that result in the zero vector.
    This implies that any linear combination of the worker computations will contain a non-zero linear combination of the random vectors. Therefore, Relay 1 cannot recover any information about the partial gradients or their linear combinations.

    For Cluster 2, the workers compute $\Tilde{\mathbf{g}}_{(2,1)}^{\pi}$, $\Tilde{\mathbf{g}}_{(2,2)}^{\pi}$ and $\Tilde{\mathbf{g}}_{(2,3)}^{\pi}$ as follows, 
    \begin{align*}
        \Tilde{\mathbf{g}}_{(2,1)}^{\pi} &\defeq \frac{1}{9} \mathbf{g}_{1} - \frac{1}{6} \mathbf{g}_{4} + \frac{2}{3} \mathbf{z}_{1} + \frac{2}{3} \mathbf{z}_{2}, \\
        \Tilde{\mathbf{g}}_{(2,2)}^{\pi} &\defeq \frac{1}{18} \mathbf{g}_{1} + \frac{1}{3} \mathbf{z}_{1} + \frac{2}{3} \mathbf{z}_{2}, \\
        \Tilde{\mathbf{g}}_{(2,3)}^{\pi} &\defeq \frac{1}{6} \mathbf{g}_{4} + \frac{2}{3} \mathbf{z}_{2}.
    \end{align*}
    Assuming $W_{(2,3)}$ as a straggler, Relay 2 receives two computations from $W_{(2,1)}$ and $W_{(2,2)}$. With these two computations, Relay 2 can interpolate the following Intra-Cluster polynomial, 
    \begin{align*}
        {\mathbf{f}_{\text{intraC}}^{\pi}}^{(2)}(x_2)&\defeq
         \mathbf{g}_{1} \frac{x_2-3}{-18} + \mathbf{g}_{4} \frac{x_2-2}{6} + \mathbf{z}_{1} \frac{x_2-3}{-3} + \mathbf{z}_{2} \frac{2}{3}.
    \end{align*}
    This is possible since ${\mathbf{f}_{\text{intraC}}^{\pi}}^{(2)}(x_2)$ is of degree 1, and the computations of the $W_{(2,1)}$ and $W_{(2,2)}$ are in fact ${\mathbf{f}_{\text{intraC}}^{\pi}}^{(2)}(\alpha_1^{(2)})$ and ${\mathbf{f}_{\text{intraC}}^{\pi}}^{(2)}(\alpha_2^{(2)})$ respectively, where $\alpha_1^{(2)}=1$ and $\alpha_2^{(2)}=2$.

    To prove privacy, we again consider the random vector component of these two computations, which can be written as 
    \begin{align*}
        \begin{bmatrix}
        \frac{2}{3} & \frac{2}{3} & 0 \\
        \frac{1}{3} & \frac{2}{3} & 0
    \end{bmatrix}
    \begin{bmatrix}
        \mathbf{z}_{1} \\
        \mathbf{z}_{2} \\
        \mathbf{z}_{3}
    \end{bmatrix},
    \end{align*}
    where each row corresponds to a worker. Again, this coefficient matrix has full row rank, which implies that Relay 2 cannot obtain any information about the partial gradients or their linear combinations. 

    The computations of workers in clusters 3 and 4 follow similarly. Therefore, we only present the Intra-Cluster polynomials for these two clusters.
    \begin{align*}
        {\mathbf{f}_{\text{intraC}}^{\pi}}^{(3)}(x_2)&\defeq
         \mathbf{g}_{2} \frac{x_2-3}{-3} + \mathbf{g}_{3} \frac{x_2-2}{16} + \mathbf{z}_{1} \frac{x_2-3}{-4} - \mathbf{z}_{3} \frac{3}{2}, \\
        {\mathbf{f}_{\text{intraC}}^{\pi}}^{(4)}(x_2)&\defeq
         \mathbf{g}_{2} \frac{x_2-3}{-1} + \mathbf{g}_{4} \frac{x_2-2}{-2} + \mathbf{z}_{2} \frac{4(x_2-3)}{9} - \mathbf{z}_{3}(4).
    \end{align*}

    After each relay recovers the corresponding Intra-Cluster polynomial, it computes its value at point $x_2=\beta_1^{(2)}=0$ and sends the result to the server. The computations of the relays are as follows,
    \begin{align*}
        \Tilde{\mathbf{g}}_{1}^{\pi} &\defeq \frac{1}{2} \mathbf{g}_{1} + \frac{3}{8} \mathbf{g}_{3} + \frac{3}{4} \mathbf{z}_{1} + \frac{2}{3} \mathbf{z}_{2} + \frac{1}{2} \mathbf{z}_{3}, \\
        \Tilde{\mathbf{g}}_{2}^{\pi} &\defeq \frac{1}{6} \mathbf{g}_{1} - \frac{1}{3} \mathbf{g}_{4} + \mathbf{z}_{1} + \frac{2}{3} \mathbf{z}_{2}, \\
        \Tilde{\mathbf{g}}_{3}^{\pi} &\defeq \mathbf{g}_{2} - \frac{1}{8} \mathbf{g}_{3} + \frac{3}{4} \mathbf{z}_{1} - \frac{3}{2} \mathbf{z}_{3}, \\
        \Tilde{\mathbf{g}}_{4}^{\pi} &\defeq 3\mathbf{g}_{2} + \mathbf{g}_{4} - \frac{4}{3} \mathbf{z}_{2} - 4 \mathbf{z}_{3},
    \end{align*}
    where $\Tilde{\mathbf{g}}_{n}^{\pi}$ is the computation of Relay $n$. In fact, these polynomials are designed to determine the coefficients of the linear encodings; for example one can easily verify that $\Tilde{\mathbf{g}}_{1}^{\pi}=2\Tilde{\mathbf{g}}_{(1,1)}^{\pi}-\Tilde{\mathbf{g}}_{(1,2)}^{\pi}$, $\Tilde{\mathbf{g}}_{2}^{\pi}=2\Tilde{\mathbf{g}}_{(2,1)}^{\pi}-\Tilde{\mathbf{g}}_{(2,2)}^{\pi}$.

    Assuming Relay 4 is a straggler, the server receives three computations from the first three relays. With these three computations, the server reconstructs the Cluster-to-Server polynomial $\mathbf{f}_{\text{C2S}}^{\pi}(x_1)$, which is
    \begin{align*}
        \mathbf{f}_{\text{C2S}}^{\pi}(x_1) &\defeq\mathbf{g}_{1} (\frac{x_1-3}{-3}) (\frac{x_1-4}{-4}) + \mathbf{g}_{2} (\frac{x_1-1}{-1}) (\frac{x_1-2}{-2}) + \mathbf{g}_{3} (\frac{x_1-2}{-2}) (\frac{x_1-4}{-4}) + \mathbf{g}_{4} (\frac{x_1-1}{-1}) (\frac{x_1-3}{-3}) \\
        &+\mathbf{z}_{1} (\frac{x_1-4}{-4})(x_1) + \mathbf{z}_{2} (\frac{x_1-3}{-3})(x_1) + \mathbf{z}_{3} (\frac{x_1-2}{-2})(x_1).
    \end{align*}
    This reconstruction is feasible because $\mathbf{f}_{\text{C2S}}^{\pi}(x_1)$ is a degree-2 polynomial and can be recovered by any three evaluations, and each message $\Tilde{\mathbf{g}}_{n}^{\pi}$ corresponds to $\mathbf{f}_{\text{C2S}}^{\pi}(x_1=n)$.
    In the end, the server recovers the aggregated  gradient vector as
    \begin{align*}
        \mathbf{g} \defeq \sum_{i=1}^4 \mathbf{g}_i = \mathbf{f}_{\text{C2S}}^{\pi}(\beta_1^{(1)}=0),
    \end{align*}
    which could also be written as $\mathbf{g}=3\Tilde{\mathbf{g}}_{1}^{\pi}-3\Tilde{\mathbf{g}}_{2}^{\pi}+\Tilde{\mathbf{g}}_{3}^{\pi}$.
\subsection{The Proposed Privacy-Preserving Hierarchical Gradient Coding Scheme} \label{sec:achprivacy}
 In this section, we describe our proposed scheme, which aims to keep the partial gradients and their linear combinations private against honest-but-curious relays. In this setting, for $r_1 \in [s_1+1:N_1-1]$, let $r'_1 = r_1+1$. We divide the partial gradients and the random vectors into $m_1$ equal-sized parts as $\mathbf{g}_k=(\mathbf{g}_k[1], ..., \mathbf{g}_k[m_1])$ and $\mathbf{z}_k=(\mathbf{z}_k[1], ..., \mathbf{z}_k[m_1])$, where $m_1=r_1-2a_1-s_1$.  The number of stragglers and adversaries among relays are denoted by $s_1$ and $a_1$, respectively. 
 Note that while the scheme tolerates stragglers among workers in Cluster $n$, denoted by $s_2^{(n)}$, it does not tolerate adversarial workers, i.e., $a_2^{(n)}=0$ for all $ n\in [N_1]$.

\subsubsection{Design of the Encoding Polynomials}
As before, the encoding process consists of two layers:  the Intra-Cluster part and the Cluster-to-Server part.
We begin by describing the design of the Cluster-to-Server encoding polynomial below.\\

\noindent\textbf{Design of the Cluster-to-Server Encoding: }  Let $\alpha_1^{(1)}, ..., \alpha_{N_1}^{(1)}$ and $\beta_1^{(1)}, ..., \beta_{m_1}^{(1)}$ be $N_1+m_1$ distinct values in  $\mathbb{F}_q$.
To design the Cluster-to-Server polynomial   $\mathbf{f}_{\text{C2S}}^{\pi}(x_1)$, we define the polynomials $\mathbf{p}_{l_1}(x_1)$ and $\mathbf{q}_{l_1}(x_1)$, corresponding to partial gradients and random vectors, respectively, as 
\begin{align}
    \mathbf{p}_{l_1}(x_1)=\sum_{k=1}^{K} \mathbf{g}_{k}[l_1] \prod^{N_1}_{\substack{i=1 \\ i:\mathcal{D}_{k}\notin\Gamma_{i}}} \frac{x_1-\alpha^{(1)}_{i}}{\beta^{(1)}_{l_1}-\alpha^{(1)}_{i}},\\
    \mathbf{q}_{l_1}(x_1)=\sum_{k=1}^{K'} \mathbf{z}_{k}[l_1] \prod^{N_1}_{\substack{i=1 \\ i:\mathbf{z}_{k}\notin\Gamma'_{i}}} \frac{x_1-\alpha^{(1)}_{i}}{\beta^{(1)}_{l_1}-\alpha^{(1)}_{i}},
\end{align}

   These polynomials have the property that $\mathbf{p}_{l_1}(\alpha_{n}^{(1)})$ and $\mathbf{q}_{l_1}(\alpha_{n}^{(1)})$ depend only on the partial gradients of datasets in $\Gamma_n$ and random vectors in $\Gamma'_n$, respectively. 
   The encoding polynomial $\mathbf{f}_{\text{C2S}}^{\pi}(x_1)$ is constructed as follows.
    \begin{align} \label{eq:outer2}
        \mathbf{f}_{\text{C2S}}^{\pi}(x_1)=&\sum_{l_1=1}^{m_1} \bigg(\mathbf{p}_{l_1}(x_1)  \prod^{m_1}_{\substack{u_1=1 \\ u_1 \neq l_1}} \frac{x_{1}-\beta^{(1)}_{u_1}}{\beta^{(1)}_{l_1}-\beta^{(1)}_{u_1}}\bigg) 
        +\sum_{l_1=1}^{m_1} \bigg(\mathbf{q}_{l_1}(x_1)  \prod^{m_1}_{\substack{u_1=1}} (x_{1}-\beta^{(1)}_{u_1})\bigg).
    \end{align}
    Note that $\mathbf{f}_{\text{C2S}}^{\pi}(x_1)$ is designed such that the second term in \eqref{eq:outer2} vanishes when evaluating at any $\beta_i^{(1)}$, for $i\in [m_1]$, i.e., 
\begin{align}
    \mathbf{f}_{\text{C2S}}^{\pi}(\beta_i^{(1)}) = \sum_{k=1}^K \mathbf{g}_k[i],
\end{align}
which ensures the correct aggregation of the 
$i$-th partition of the gradients at the server.

The degree of $\mathbf{f}_{\text{C2S}}^{\pi}(x_1)$ is at most $N_1-2a_1-s_1-1$, since 
\begin{itemize}
    \item $\deg(\mathbf{p}_{l_1}(x_1))=N_1-r_1$,  and it is multiplied by a degree $m_1-1=r_1-2a_1-s_1-1$ polynomial.
    \item $\deg(\mathbf{q}_{l_1}(x_1))=N_1-(r_1+1)$, and it is multiplied by a degree $m_1=r_1-2a_1-s_1$ polynomial.
\end{itemize}
The task of Relay $n$ is to compute and send 
\begin{align*}
\Tilde{\mathbf{g}}_n^{\pi}\defeq\mathbf{f}_{\text{C2S}}^{\pi}(\alpha_n^{(1)}) =&\sum_{l_1=1}^{m_1} \bigg(\mathbf{p}_{l_1}(\alpha_n^{(1)})  \prod^{m_1}_{\substack{u_1=1 \\ u_1 \neq l_1}} \frac{\alpha_n^{(1)}-\beta^{(1)}_{u_1}}{\beta^{(1)}_{l_1}-\beta^{(1)}_{u_1}}\bigg) 
        +&\sum_{l_1=1}^{m_1} \bigg(\mathbf{q}_{l_1}(\alpha_n^{(1)})  \prod^{m_1}_{\substack{u_1=1}} (\alpha_n^{(1)}-\beta^{(1)}_{u_1})\bigg).
\end{align*}
The length of $\Tilde{\mathbf{g}}_n^{\pi}$ is $\frac{d}{m_1}$, as it is a linear combination of vectors of that same length. Thus, the communication cost at this stage is $C_1^{\pi}=\frac{d}{m_1}$.

\noindent   \textbf{Design of the Intra-Cluster Encoding: } For each Cluster $n\in [N_1]$, we need a separate Intra-Cluster encoding polynomial accustomed to the corresponding Cluster-to-Server encoding $\Tilde{\mathbf{g}}_n$.  First, we need to impose certain constraints on the random vectors assigned to the workers, which we describe through the assignment and the total number of random vectors $K'$ as follows.

\textbf{Random Vectors Assignments. } To ensure proper assignments, we require that  $|\Gamma'_n|\geq N_2^{(n)}-s_2^{(n)}$. We set $r_{2,n}=r'_{2,n}=s_2^{(n)}+1$, which leads to $m_2^{(n)}=1$. This means that in Cluster $n$ we need  at least $N_2^{(n)}-s_2^{(n)}$ distinct random vectors, each repeated at least  $s_2^{(n)}+1$ times among the workers. Additionally, we assign random vectors such that there exists one with repetition 
$s_2^{(n)}+1$, another with repetition $s_2^{(n)}+2$, and so on, up to a vector with repetition  $N_2^{(n)}$.

\textbf{Requirements on $K'$. }
To determine the value of $K'$,  we have that the workers in Cluster $n$ must be assigned at least  $N_2^{(n)}-s_2^{(n)}$ distinct random vectors, which implies 
\begin{align}
   K'\geq \max_{n\in [N_1]} (N_2^{(n)}-s_2^{(n)}). 
\end{align}
But it may not be enough to just put $K'= \max_{n\in [N_1]} (N_2^{(n)}-s_2^{(n)})$ since the assignment should be such that the minimum number of repetitions of all random vectors among clusters is equal to $r'_1=r_1+1$. On the other hand, the total number of random vector copies assigned across all clusters is $\sum_{n\in [N_1]}|\Gamma'_n|$, and since each $\mathbf{z}_k$ is repeated at least $r'_1=r_1+1$ times among clusters, for any assignment we naturally get that
\begin{align}
    K' \leq \frac{\sum_{n\in [N_1]}|\Gamma'_n|}{r_1+1}.
\end{align}


\begin{claim}
There always exist parameters $K'$ and $|\Gamma'_n|$, for all $n \in [N_1]$  that satisfy the above requirements for any given values of $r_1, N_2^{(n)}$ and, $s_2^{(n)}$.

\end{claim}
\begin{proof}
    We propose an algorithm that meets the above requirements with $K' = \frac{\sum_{n\in [N_1]}|\Gamma'_n|}{r_1+1}$ and $K'\geq \max_{n\in [N_1]} (N_2^{(n)}-s_2^{(n)})$ by designing $K'$ and $\Gamma'_n$, for all $ n\in [N_1]$.

The algorithm consists of multiple steps and in each step, the values $K'$ and $\Gamma'_1, ..., \Gamma'_{N_1}$ get updated. Let ${K'}^{(i)}$ and ${\Gamma'}^{(i)}_1, \cdots, {\Gamma'}^{(i)}_{N_1}$ denote these values at stage $i$ of the algorithm. 
   Assume, without loss of generality that the clusters are sorted based on decreasing order for values $N_2^{(n)}-s_2^{(n)}, \forall n\in [N_1]$. We initialize with
 ${K'}^{(0)}=0$ and set all initial assignment sets, i.e., ${\Gamma'}^{(0)}_1, \cdots, {\Gamma'}^{(0)}_{N_1}$, to the empty set $\{\}$.
  In the first stage, we set ${K'}^{(1)}=N_2^{(1)}-s_2^{(1)}$ and assign $\mathbf{z}_1, \cdots, \mathbf{z}_{K'^{(1)}}$ to the first cluster, leading to the updated assignment set ${\Gamma'}^{(1)}_1$; this satisfies $|{\Gamma'}^{(1)}_1|\geq N_2^{(1)}-s_2^{(1)}$. Each of the random vectors  is then randomly assigned to other clusters to ensure that their repetition reaches $r_1+1$. This leads to the updated assignment sets ${\Gamma'}^{(1)}_1, \cdots, {\Gamma'}^{(1)}_{N_1}$. 

  In the second stage, we proceed to Cluster 2.
 If $|{\Gamma'}^{(1)}_2|< N_2^{(2)}-s_2^{(2)}$, 
 we add additional random vectors such that
 ${K'}^{(2)}= {K'}^{(1)}+(N_2^{(2)}-s_2^{(2)}-|{\Gamma'}^{(1)}_2|)$ and assign $\mathbf{z}_{{K'}^{(1)}+1}, \cdots, \mathbf{z}_{{K'}^{(2)}}$ to Cluster 2. This leads to ${\Gamma'}^{(2)}_2$, which satisfies $|{\Gamma'}^{(2)}_2|\geq N_2^{(2)}-s_2^{(2)}$.
 These new vectors are then assigned randomly to other clusters until each one is repeated $r_1+1$ times. This results in updated sets ${\Gamma'}^{(2)}_1, \cdots, {\Gamma'}^{(2)}_{N_1}$.

 We repeat this process for each cluster. At stage $i$, we update the assignment sets ${\Gamma'}^{(i)}_1, \cdots, {\Gamma'}^{(i)}_{N_1}$ and the current total number of random  vectors ${K'}^{(i)}$. The algorithm terminates after stage 
 $N_1$, where we set
  $K'={K'}^{(N_1)}$ and ${\Gamma'}_{n}={\Gamma'}^{(N_1)}_{n}$, for all $n \in [N_1]$. 
  By construction, the conditions 
   $|\Gamma'_n|\geq N_2^{(n)}-s_2^{(n)}$, for all  $n\in [N_1]$, and $(r_1+1)K'=\sum_{n\in [N_1]}|\Gamma'_n|$ hold, since each $\mathbf{z}_k$ is repeated exactly $r_1+1$ times.
 Additionally, we have $K'\geq \max_{n\in [N_1]} (N_2^{(n)}-s_2^{(n)}) = N_2^{(1)}-s_2^{(1)}$.
\end{proof} 

Now we can proceed to the construction of the Intra-Cluster encoding polynomial.
Let $\alpha_1^{(2)}, ..., \alpha_{N_2^{(n)}}^{(2)}$ and $\beta_1^{(2)}$ be $N_2^{(n)}+1$ distinct values chosen from $\mathbb{F}_q$.
We first define the polynomials $\mathbf{p}^{(n)}(x_2)$ and $\mathbf{q}^{(n)}(x_2)$ as follows.
    \begin{align*}
        &\mathbf{p}^{(n)}(x_2)\defeq  \sum_{l_1=1}^{m_1}\hspace{-1mm}\left(\sum_{k=1}^{K} \mathbf{g}_{k}[l_1]\prod^{N_1}_{\substack{i=1 \\ i:\mathcal{D}_{k}\notin\Gamma_{i}}} \frac{\alpha_n^{(1)}-\alpha^{(1)}_{i}}{\beta^{(1)}_{l_1}-\alpha^{(1)}_{i}}
        \prod^{N_2^{(n)}}_{\substack{j=1 \\ j:\mathcal{D}_{k}\notin\Gamma_{(n,j)}}} \hspace{-2mm}\frac{x_2-\alpha^{(2)}_{j}}{\beta^{(2)}_1-\alpha^{(2)}_{j}} \right)  \prod^{m_1}_{\substack{u_1=1 \\ u_1 \neq l_1}} \frac{\alpha_1^{(1)}-\beta^{(1)}_{u_1}}{\beta^{(1)}_{l_1}-\beta^{(1)}_{u_1}}, \\
        &\mathbf{q}^{(n)}(x_2)\defeq  \sum_{l_1=1}^{m_1}\hspace{-1mm}\left(\sum_{k=1}^{K'} \mathbf{z}_{k}[l_1]\prod^{N_1}_{\substack{i=1 \\ i:\mathbf{z}_{k}\notin\Gamma'_{i}}} \frac{\alpha_n^{(1)}-\alpha^{(1)}_{i}}{\beta^{(1)}_{l_1}-\alpha^{(1)}_{i}}
        \prod^{N_2^{(n)}}_{\substack{j=1 \\ j:\mathbf{z}_{k}\notin\Gamma'_{(n,j)}}} \hspace{-2mm}\frac{x_2-\alpha^{(2)}_{j}}{\beta^{(2)}_{1}-\alpha^{(2)}_{j}} \right)  \prod^{m_1}_{\substack{u_1=1}} (\alpha_1^{(1)}-\beta^{(1)}_{u_1}).
    \end{align*}
The Intra-Cluster encoding polynomial for Cluster $n$, denoted by ${\mathbf{f}_{\text{intraC}}^{\pi}}^{(n)}(x_2)$, is defined as follows. 
    \begin{align} \label{eq:fintraprivacy}
        {\mathbf{f}_{\text{intraC}}^{\pi}}^{(n)}(x_2) =
        \mathbf{p}^{(n)}(x_2) + \mathbf{q}^{(n)}(x_2).
    \end{align}
Note that ${\mathbf{f}_{\text{intraC}}^{\pi}}^{(n)}(x_2)$ has  degree $N_2^{(n)}-s_2^{(n)}-1$, since both $\mathbf{p}^{(n)}(x_2)$ and $\mathbf{q}^{(n)}(x_2)$ are polynomials of degree $N_2^{(n)}-r_{2,n}=N_2^{(n)}-s_2^{(n)}-1$.

The task of each worker $W_{(n,j)}$ is to send a linear combination of its computed partial gradients and assigned random vectors, denoted by $\Tilde{\mathbf{g}}_{(n,j)}^{\pi}$, to Relay $n$  which is, in fact, the evaluation of ${\mathbf{f}_{\text{intraC}}^{\pi}}^{(n)}(x_2)$ at $\alpha_j^{(2)}$; that is,
    \begin{align} \label{eq:workertaskprivacy}
        \Tilde{\mathbf{g}}_{(n,j)}^{\pi} \defeq {\mathbf{f}_{\text{intraC}}^{\pi}}^{(n)}(\alpha_j^{(2)})= \mathbf{p}^{(n)}(\alpha_j^{(2)}) + \mathbf{q}^{(n)}(\alpha_j^{(2)}),
    \end{align}
The length of this vector is $\frac{d}{m_1}$, as it results from a linear combination of partitions of partial gradients and random vectors, each of length $\frac{d}{m_1}$; therefore, $C_2^{\pi^{(n)}}=\frac{d}{m_1}$.

\subsubsection{Privacy-Preserving Hierarchical Gradient Coding: Step-by-Step Algorithm}
   With the provided description, the proposed privacy-preserving scheme proceeds as follows:
\begin{enumerate}
    \item \textbf{Worker-to-Relay Transmission.} $W_{(n,j)}$ computes  $\Tilde{\mathbf{g}}_{(n,j)}^{\pi}$ as defined in \eqref{eq:workertaskprivacy}, and sends the result to Relay $n$. 
    \item \textbf{Calculating Intra-Cluster Polynomial by the Relay.} Upon receiving $\Tilde{\mathbf{g}}_{(n,j)}^{\pi}$ from non-straggling workers, Relay $n$ recovers ${\mathbf{f}_{\text{intraC}}^{\pi}}^{(n)}(x_2)$, which is of degree $N_2^{(n)}-s_2^{(n)}-1$. 
    \item \textbf{Relay-to-Server Transmission.} Relay $n$ computes $\Tilde{\mathbf{g}}_{n}^{\pi} = \mathbf{f}_{\text{C2S}}^{\pi}(\alpha_n^{(1)}) = {\mathbf{f}_{\text{intraC}}^{\pi}}^{(n)}(\beta_1^{(2)})$ and sends it to the server.
    \item  \textbf{Calculating Cluster-to-Server Polynomial by the Server.} The server, having received $\Tilde{\mathbf{g}}_{n}^{\pi}$ from non-straggling relays, recovers $\mathbf{f}_{\text{C2S}}^{\pi}(x_1)$,  which is of degree  $N_1-2a_1-s_1-1$. 
    \item \textbf{Final Gradient Recovery.}  The server calculates the aggregated partial gradients ${\mathbf{g}} = ({\mathbf{g}}[1],...,{\mathbf{g}}[m_1]) = (\mathbf{f}_{\text{C2S}}^{\pi}(\beta_1^{(1)}), ..., \mathbf{f}_{\text{C2S}}^{\pi}(\beta_{m_1}^{(1)}))$.
    
\end{enumerate}
 

\begin{remark}
    The proposed privacy-preserving hierarchical gradient coding scheme operates for values of $r_1 \in [s_1+1:N_1-1]$ and $r_{2,n}=s_2^{(n)}+1$,  rather than the full possible ranges $r_1 \in [s_1+1:N_1]$ and $r_{2,n} \in [s_2^{(n)}+1:N_2^{(n)}]$. Exploring these broader ranges is left for future work.
\end{remark}

\begin{remark}
    The proposed scheme achieves the optimal communication loads under the specified system parameters. It demonstrates that enforcing the privacy constraint does not require increasing the communication load compared to the non-private setting.
\end{remark}

\begin{remark}
In the proposed privacy-preserving scheme, we assume that in Cluster $n$, exactly $s^{(2)}_{n}$ links fail to deliver any computation to the corresponding relay. The information-theoretic definition of privacy in this setting, as defined in \eqref{eq:privacy}, also reflects this assumption.
\end{remark}
   

\subsection{Proof of Privacy} \label{sec:proofprivacy}

Consider Cluster $n$. The relay receives $N_2^{(n)}-s_2^{(n)}$ computations from the non-straggling workers. The computation of $W_{(n,j)}$, denoted by $\Tilde{\mathbf{g}}_{(n,j)}^{\pi}$, and computed as in \eqref{eq:workerencprivacy}, corresponds to the evaluation of the polynomial  ${\mathbf{f}_{\text{intraC}}^{\pi}}^{(n)}(x_2)$ at point $x_2=\alpha_j^{(2)}$. This polynomial can be expressed as a vector-matrix product,
\begin{align}
    {\mathbf{f}_{\text{intraC}}^{\pi}}^{(n)}(x_2) = [\mathbf{b}(x_2) \text{ }\mathbf{a}(x_2)] \begin{bmatrix}
               \mathbf{g}_{1}[1] \\
               \vdots \\
               \mathbf{g}_{K}[m_1] \\
               \mathbf{z}_{1}[1] \\
               \vdots \\
               \mathbf{z}_{K'}[m_1] \\
            \end{bmatrix},
\end{align}
where $\mathbf{b}(x_2)= [b_1(x_2), ..., b_{Km_1}(x_2)]$ is the coefficient vector of length $Km_1$ associated with the partial gradients, and $\mathbf{a}(x_2)= [a_1(x_2), ..., a_{K'm_1}(x_2)]$ is the coefficient vector of length $K'm_1$ associated with the random vectors. 

Assuming, without loss of generality, that the first $N_2^{(n)}-s_2^{(n)}$ workers are the non-stragglers, we can write
\begin{align*}
    &\begin{bmatrix}
        \Tilde{\mathbf{g}}_{(n,1)}^{\pi} \\
        \vdots \\
        \Tilde{\mathbf{g}}_{(n,N_2^{(n)}-s_2^{(n)})}^{\pi} \\
    \end{bmatrix} =  \begin{bmatrix}
        \mathbf{b}(\alpha_1^{(2)}) & \mathbf{a}(\alpha_1^{(2)}) \\
        \vdots & \vdots \\
        \mathbf{b}(\alpha_{N_2^{(n)}-s_2^{(n)}}^{(2)}) & \mathbf{a}(\alpha_{N_2^{(n)}-s_2^{(n)}}^{(2)}) \\
    \end{bmatrix}
    \begin{bmatrix}
        \mathbf{g}_{1}[1] \\
        \vdots \\
        \mathbf{g}_{K}[m_1] \\
        \mathbf{z}_{1}[1] \\
        \vdots \\
        \mathbf{z}_{K'}[m_1] \\
    \end{bmatrix}.
\end{align*}
This can be separated into two matrix products
\begin{align*}
    &\begin{bmatrix}
        \Tilde{\mathbf{g}}_{(n,1)}^{\pi} \\
        \vdots \\
        \Tilde{\mathbf{g}}_{(n,N_2^{(n)}-s_2^{(n)})}^{\pi} \\
    \end{bmatrix} = 
    \begin{bmatrix}
        \mathbf{b}(\alpha_1^{(2)})\\
        \vdots \\
        \mathbf{b}(\alpha_{N_2^{(n)}-s_2^{(n)}}^{(2)}) \\
    \end{bmatrix}
    \begin{bmatrix}
        \mathbf{g}_{1}[1] \\
        \vdots \\
        \mathbf{g}_{K}[m_1] \\
    \end{bmatrix}  +
    \begin{bmatrix}
        a_1(\alpha_{1}^{(2)})& \cdots & a_{K'm_1}(\alpha_{1}^{(2)})\\
        &\vdots & \\
        a_1(\alpha^{(2)}_{N_2^{(n)}-s_2^{(n)}})& \cdots & a_{K'm_1}(\alpha^{(2)}_{N_2^{(n)}-s_2^{(n)}})\\
    \end{bmatrix}
    \begin{bmatrix}
        \mathbf{z}_{1}[1] \\
        \vdots \\
        \mathbf{z}_{K'}[m_1] \\
    \end{bmatrix}
\end{align*}

To satisfy the privacy constraint in \eqref{eq:privacy}, it suffices to show that the matrix of coefficients corresponding to the random vectors—denoted by $\mathbf{A}$—has full row rank. This ensures that the contribution of the random vectors fully masks the private information.

Each $a_i(x_2)$ can be written as
\begin{align}
    a_i(x_2) = \mathbf{x} \mathbf{a}_i, 
\end{align}
where $\mathbf{x} = (1,x_2,x_2^2,\cdots,x_2^{N_2^{(n)}-s_2^{(n)}-1})$ is a vector of monomials, and $\mathbf{a}_i~=\begin{bmatrix}
        a_{i,1} \\
        \vdots \\
        a_{i,N_2^{(n)}-s_2^{(n)}} \\
    \end{bmatrix}$ is the coefficient vector for the $i$-th partition of the random vectors, which is of length $N_2^{(n)}-s_2^{(n)}$. This is because ${\mathbf{f}_{\text{intraC}}^{\pi}}^{(n)}(x_2)$ defined in \eqref{eq:fintraprivacy} is of degree $N_2^{(n)}-s_2^{(n)}-1$. Using this, we can further decompose $\mathbf{A}$ as 
\begin{align*}
    &\mathbf{A} = \begin{bmatrix}
        1& \alpha_1^{(2)} & (\alpha_1^{(2)})^{2} & \cdots & (\alpha_1^{(2)})^{N_2^{(n)}-s_2^{(n)}-1}\\
        1& \alpha_2^{(2)} & (\alpha_2^{(2)})^{2} & \cdots & (\alpha_2^{(2)})^{N_2^{(n)}-s_2^{(n)}-1}\\
        \vdots&\vdots & \vdots &\vdots&\vdots \\
        1& \alpha_{N_2^{(n)}-s_2^{(n)}}^{(2)} & (\alpha_{N_2^{(n)}-s_2^{(n)}}^{(2)})^{2} & \cdots & (\alpha_{N_2^{(n)}\hspace{-1mm}-s_2^{(n)}}^{(2)})^{N_2^{(n)}\hspace{-1mm}-s_2^{(n)}\hspace{-1mm}-1}\\
    \end{bmatrix}[\mathbf{a}_1, \mathbf{a}_2, \cdots, \mathbf{a}_{K'm_1}].
\end{align*}
The first matrix is the transpose of a Vandermonde matrix with dimension $N_2^{(n)}-s_2^{(n)} \times N_2^{(n)}-s_2^{(n)}$, which is full rank since the evaluation points $\{\alpha_j^{(2)}\}_{j=1}^{N_2^{(n)}-s_2^{(n)}}$ are distinct. Thus, to prove that $\mathbf{A}$ has full row rank, it suffices to  find $N_2^{(n)}-s_2^{(n)}$ columns from $[\mathbf{a}_1, \mathbf{a}_2, \cdots, \mathbf{a}_{K'm_1}]$, or in other words $N_2^{(n)}-s_2^{(n)}$ coefficient vectors $\mathbf{a}_i$ that are linearly independent. 

Remember that in the proposed scheme, we have $|\Gamma'_n|\geq N_2^{(n)}-s_2^{(n)}$; meaning there are at least $N_2^{(n)}-s_2^{(n)}$ distinct random vectors in the cluster. Moreover, there is one random vector with repetition $s_2^{(n)}+1$, one with repetition $s_2^{(n)}+2$, and so on, up to  a random vector with repetition $N_2^{(n)}$. Without loss of generality, assume these random vectors are labeled as $\mathbf{z}_1, \mathbf{z}_2, \cdots, \mathbf{z}_{N_2^{(n)}-s_2^{(n)}}$.
From the construction of ${\mathbf{f}_{\text{intraC}}^{\pi}}^{(n)}(x_2)$ in \eqref{eq:fintraprivacy}, we infer that the degree of the coefficient of $\mathbf{z}_{k}[l_1]$ is equal to the number of workers $(n,j)$ in Cluster $n$ for which $\mathbf{z}_{k} \notin \Gamma'_{(n,j)}$. Therefore since the number of workers $(n,j)$ for which $\mathbf{z}_{1} \in \Gamma'_{(n,j)}$ is $s_2^{(n)}+1$, the coefficient polynomial for partitions of $\mathbf{z}_1$ would be of degree $N_2^{(n)}-s_2^{(n)}-1$. This means the coefficient vectors $\mathbf{a}_1, \cdots, \mathbf{a}_{m_1}$ corresponding to $\mathbf{z}_1[1], \cdots, \mathbf{z}_1[m_1]$, would have non-zero values for the $(N_2^{(n)}-s_2^{(n)})$-th position. Similarly, for partitions of $\mathbf{z}_2$, since the number of repetitions is $s_2^{(n)}+2$, the coefficients would be of degree $N_2^{(n)}-s_2^{(n)}-2$. This means the coefficient vectors $\mathbf{a}_{m_1+1}, \cdots, \mathbf{a}_{2m_1}$ corresponding to $\mathbf{z}_2[1], \cdots, \mathbf{z}_2[m_1]$, would have zero values for the $(N_2^{(n)}-s_2^{(n)})$-th position but non-zero values for $(N_2^{(n)}-s_2^{(n)}-1)$-th position. Proceeding similarly, coefficient vectors $\mathbf{a}_{2m_1+1}, \cdots, \mathbf{a}_{3m_1}$ corresponding to $\mathbf{z}_3[1], \cdots, \mathbf{z}_3[m_1]$, would have zero values for the $(N_2^{(n)}-s_2^{(n)})$-th and $(N_2^{(n)}-s_2^{(n)}-1)$-th position, but non-zero values for the $(N_2^{(n)}-s_2^{(n)}-2)$-th position. This continues until we have that the coefficient vectors $\mathbf{a}_{(N_2^{(n)}-s_2^{(n)}-1)m_1+1}, \cdots, \mathbf{a}_{(N_2^{(n)}-s_2^{(n)})m_1}$ corresponding to $\mathbf{z}_{N_2^{(n)}-s_2^{(n)}}[1], \cdots, \mathbf{z}_{N_2^{(n)}-s_2^{(n)}}[m_1]$, have non-zero values for the first position, but zero values elsewhere. To find $N_2^{(n)}-s_2^{(n)}$ coefficient vectors which are linearly independent we can consider anything of the form
\begin{align}
    \left[\mathbf{a}_{i_1}, \mathbf{a}_{m_1+i_2}, \cdots, \mathbf{a}_{(N_2^{(n)}-s_2^{(n)}-1)m_1+i_{N_2^{(n)}-s_2^{(n)}}}\right],
\end{align}
where $i_1, i_2,\cdots, i_{N_2^{(n)}-s_2^{(n)}} \in [1:m_1]$. This leads to an upper triangular matrix with non-zero values on the diagonal, therefore being linearly independent. 
\section{Conclusion}
In this paper, we studied the hierarchical gradient coding problem with the goal of increasing the scalability of gradient coding by reducing the bandwidth requirement at the server. We first derived the optimal communication-computation trade-off by proposing a converse and an achievable scheme. Notice that the tightness of the converse in under the assumption of linear encoding functions. The achievable scheme is constructed through the careful design of two layers of polynomials, which determine the encoding and decoding tasks at the nodes. Each worker's task is designed not only to minimize the worker-to-relay communication load but also to enable each relay to minimize the relay-to-server communication. Each relay is required to process the data rather than simply forwarding the messages received from the workers. This processing is enabled by the layered design of the polynomials. We concluded that this new hierarchical topology can significantly reduce the bandwidth requirement at the server. We then extended the problem to a privacy-preserving hierarchical gradient coding scenario, where the information on partial gradients is kept private from the relays. Interestingly, we showed that, for a wide range of system parameters, this is possible without increasing the communication cost. The methodology for determining the encoding and decoding functions is again based on designing polynomials using partial gradients and random vectors. This ensures that the messages sent by the workers to the relays always include some randomly generated vectors—privately generated by the workers—such that the relays cannot infer any information about the partial gradients upon receiving the workers' messages.

\bibliographystyle{IEEEtran}
 \bibliography{references}
 	

\end{document}